\documentclass[12pt]{TD-CJS}
\pdfoutput=1

\usepackage{latexsym}
\usepackage{amsmath}
\usepackage{amsfonts}
\usepackage{amssymb}
\usepackage{psfrag}
\usepackage{graphicx}
\usepackage{Cjsref}

\def\AIC{\textsc{aic}}
\def\T{{ \mathrm{\scriptscriptstyle T} }}

\makeatletter
\DeclareMathOperator*{\argmin}{arg\,min}
\makeatother

\makeatletter
\let\oldnorm\norm
\def\norm{\@ifstar{\oldnorm}{\oldnorm*}}
\makeatother

\DeclareMathOperator{\rank}{rank}
\DeclareMathOperator{\cov}{cov}
\DeclareMathOperator{\diag}{diag}
\DeclareMathOperator{\trace}{tr}

\renewcommand{\hat}{\widehat}
\renewcommand{\tilde}{\widetilde}

\newtheorem{condition}{Condition}
\newtheorem{algorithm}{Algorithm}
\newtheorem{example}{Example}
\newtheorem{thm}{Theorem}

\begin{document}
\firstpage{1}
\lastpage{26}
\jvol{}
\issue{}
\jyear{2017}

\renewcommand{\eqref}[1]{(\ref{#1})}
\newcommand{\mb}[1]{\mathbf{#1}}
\newcommand{\mbb}[1]{\mathbb{#1}}
\newcommand{\mt}[1]{\mathrm{#1}}
\newcommand{\rv}{random variable}

\title[]{High-dimensional covariance matrix estimation using a low-rank and diagonal decomposition}
\author{Yilei Wu \authorref{1}}
\author{Yingli Qin \authorref{1}}
\author{Mu Zhu \authorref{1}}
\affiliation[1]{The University of Waterloo} 

\startabstract{%
\keywords{
\KWDtitle{Key words and phrases}
Akaike information criterion\sep coordinate descent\sep consistency\sep eigen-decomposition\sep Kullback--Leibler loss\sep log-determinant semi-definite programming\sep Markowitz portfolio selection.}
\begin{abstract}
\abstractsection{}
\abstractsection{Abstract}
We study high-dimensional covariance/precision matrix estimation under the assumption that the covariance/precision matrix can be decomposed into a low-rank component $L$ and a diagonal component $D$. The rank of $L$ can either be chosen to be small or controlled by a penalty function. Under moderate conditions on the population covariance/precision matrix itself and on the penalty function, we prove some consistency results for our estimators. A blockwise coordinate descent algorithm, which iteratively updates $L$ and $D$, is then proposed to obtain the estimator in practice. Finally, various numerical experiments are presented: using simulated data, we show that our estimator performs quite well in terms of the Kullback--Leibler loss; using stock return data, we show that our method can be applied to obtain enhanced solutions to the Markowitz portfolio selection problem.
\end{abstract}}
\makechaptertitle


\section{INTRODUCTION}                          
\label{sec:intro}
\subsection{Brief review}
\label{sec:review}
Statistical inference in high-dimensional settings, where the data dimension $p$ is close to or larger than the sample size $n$, has been an intriguing area of research. Applications include gene expression data analysis, fMRI analysis, climate studies, financial economics, and many others. Estimating large covariance matrices is an essential part of high-dimensional data analysis because of the ubiquity of covariance matrices in statistical procedures, such as discriminant analysis and hypothesis testing. However, in high dimensions, the sample covariance matrix $S$ is no longer an accurate estimator of the population covariance matrix; it may not even be positive definite. To overcome these difficulties, researchers have been developing new methods.

The simplest solution is to use a scaled identity or a diagonal matrix as a substitute for the sample covariance matrix. It is well-known that the sample covariance matrices tend to overestimate the large eigenvalues and underestimate the small eigenvalues of the population covariance matrix; this bias can be corrected by shrinking the sample covariance matrix towards a scaled identity matrix, e.g., $\trace(S)I_p$ \citep{friedman1989regularized}. An optimal weight for the convex linear combination between the sample covariance matrix and the identity matrix has been proposed and studied by \cite{ledoit2004well}. Ignoring the correlations and preserving only the diagonal part of $S$ is a long-established practice in the high-dimensional classification, often referred to as the independence rule or the ``naive Bayes classifier''; it has been demonstrated to outperform Fisher's linear discriminant rule under certain conditions \citep{dudoit2002comparison,bickel2004some,fan2008high}. 

Apart from the scaled identity matrix and the diagonal matrix, other structured estimators have also been proposed. Methods such as banding \citep{bickel2008regularized} and tapering \citep{furrer2007estimation} are useful when the covariates have a natural ordering \citep{rothman2008sparse}. \cite{cai2013optimal} studied banding and tapering estimators in estimating large Toeplitz covariance matrices, which arise in the analysis of stationary time series. \cite{ourQDA} exploited the compound symmetry structure to facilitate quadratic discriminant analysis (QDA) in high dimensions.    

Another popular assumption is sparse covariance or precision matrices. Sparse covariance matrix estimators can be obtained by either thresholding or regularization. Thresholding has been studied by \cite{bickel2008regularized} and \cite{cai2011adaptive}, and applied in discriminant analysis by  \cite{shao2011sparse} and \cite{li2014sparse}. To encourage sparsity, \cite{rothman2012positive} and \cite{xue2012positive} imposed lasso-type penalties on the covariance matrix. Sparsity is a good assumption for the precision matrix in many applications, e.g., for Gaussian data zeros in the precision matrix suggest conditional independence; it can be achieved directly by imposing an $\ell_1$ penalization on the precision matrix  \citep{yuan2007model,rothman2008sparse,banerjee2008model,friedman2008sparse, lam2009sparsistency,cai2011constrained} or indirectly through regularized regression \citep{meinshausen2006high,rocha2008path,yuan2010high,sun2013sparse}.

In the context of high-dimensional data analysis, it is reasonable to assume that the variance of the observed data can be explained by a small number of latent factors; thus, factor models can be applied to reduce the number of parameters in covariance matrix estimation, too. Assuming observable factors and independent error terms, \cite{fan2008high2} proposed a covariance matrix estimator by estimating the loading matrix with regression and the covariance matrix of the error terms with a diagonal matrix. This method was generalized by \cite{fan2011high} so that the error covariance was not necessarily diagonal, but it was assumed to be sparse and estimated with thresholding techniques. \cite{fan2013large} then considered the case where the factors are unobservable. Assuming the number of latent factors ($k$) to be known, they performed PCA on the sample covariance matrix, kept the first $k$ principal components to estimate the covariance matrix of the latent factors, and thresholded the remaining principal components to estimate a sparse covariance matrix for the error terms. 

A related matrix structure is called ``spiked covariance matrix'', that is, the covariance matrix has only a few eigenvalues greater than one and can be decomposed into a low-rank matrix plus an identity matrix \citep{johnstone2001distribution}. \cite{cai2015optimal} proposed a sparse spiked covariance matrix estimator. In addition to the spiked structure, they assumed that the matrix spanned by the eigenvectors of the low-rank component has a small number of nonzero rows, which in turn constrains the covariance matrix to have a small number of rows and columns containing nonzero off-diagonal entries. 

\cite{chandrasekaran2010latent} proposed a latent variable method for Gaussian graphical model selection, based on the conditional independence interpretation of zero off-diagonals in the precision matrix. Assuming the observable and latent variables are jointly distributed as Gaussian, they showed that, if one assumes (i) the conditional precision matrix of the observables given the latent factors is sparse and (ii) the number of latent factors is small, then the marginal precision matrix of the observables must consist of a sparse component plus a low-rank component. The authors then considered a penalized likelihood approach to estimate such a marginal precision matrix, using the $\ell_1$-norm to regularize the sparse component and the nuclear-norm to regularize the low-rank component. They also derived some consistency results for their estimator in the operator norm. \cite{taeb2016interpreting} extended this framework to allow the incorporation of covariates. 

A comprehensive review has been provided by \cite{cai2016estimating}, in which they also compared some of the aforementioned methods in terms of their respective convergence rates. 

\subsection{Summary of this paper}

In this paper, we make the explicit structural assumption that the population covariance/precision matrix can be decomposed into a low-rank plus a diagonal matrix, in order to facilitate the estimation of large covariance/precision matrices in high dimensions. In Section \ref{sec:decomposition}, we discuss this main model assumption in more detail. 

While this model assumption is similar (but not identical) to some of the works reviewed in Section~\ref{sec:review}, the main difference is that we do not rely on nuclear norm regularization to promote low-rank-ness; instead, we directly impose a penalty on the matrix rank itself. 
In Section \ref{sec:FR} and Section \ref{sec:rank-penalty}, we present estimators of the covariance/precision matrix under this model assumption, and show that estimation consistency can be achieved with a proper choice of the penalty function. 

As is often the case, our estimators are characterized, or defined, as solutions to various optimization problems. In Section \ref{sec:algorithm}, we describe an efficient blockwise coordinate descent algorithm for solving the main optimization problem. In particular, given the low-rank component, the diagonal component can be obtained by solving a relatively cheap log-determinant semi-definite program; given the diagonal component, the low-rank component actually can be obtained analytically. Since optimization with nuclear-norm constraints is still computationally burdensome for large matrices, we think our approach, which avoids nuclear-norm regularization, can be especially attractive. 

In Section \ref{sec:numerical} and Section \ref{sec:realdata}, we demonstrate the performances of our method with various simulations and an analysis of some real financial data. All proofs are relegated to the supplementary material.

\section{LOW-RANK AND DIAGONAL MATRIX DECOMPOSITION}
\label{sec:decomposition}
\subsection{Notations} 
We use $\mathbf{R}^{p_1 \times p_2}$ to denote the set of $p_1 \times p_2$ matrices, $\mathbf{S}^p$ to denote the set of symmetric $p \times p$ matrices, $\mathbf{S}_{+}^p$ to denote the subset of matrices $\subset \mathbf{S}^p$ which are positive semi-definite, and $\mathbf{S}_{++}^p$ to denote the subset of those which are strictly positive definite. Sometimes, another superscript is added to denote a restriction on the rank, for example, $\mathbf{S}^{p,r}$ is used to denote the subset of matrices in $\mathbf{S}^p$ with rank $\leq r$, and likewise for $\mathbf{S}^{p,r}_+$, $\mathbf{S}^{p,r}_{++}$. For the corresponding sets of diagonal matrices, we replace $\mathbf{S}$ with $\mathbf{D}$ , e.g., $\mathbf{D}^p$, $\mathbf{D}_{+}^p$, and $\mathbf{D}_{++}^p$.

For any $A \in \mathbf{S}^{p}$, we use $\trace(A)$ to denote its trace, $|A|$ to denote its determinant, and $\lambda_{\max}(A), \lambda_{\min}(A)$ to denote its largest and smallest eigenvalues. Furthermore, we use $\|A\|_F=\{\trace(A^{\T}A)\}^{1/2}$ to denote its Frobenius norm, $\|A\|_*=\trace\{(A^{\T}A)^{1/2}\}$ to denote its nuclear norm (which is equivalent to the sum of its singular values), $\|A\|_{op}=\{\lambda_{\max}(AA^{\T})\}^{1/2}$ to denote its operator norm, and $\|A\|_1=\sum_{i,j}|A_{ij}|$ to denote its $\ell_1$ norm.

\subsection{Problem set-up and model assumption}
Consider a random sample $X=\left(x_1,\ldots,x_n\right)$, in which $x_1,\ldots,x_n$ are independently and identically distributed $p$-variate random vectors from the multivariate normal distribution with population mean $0$ and population covariance matrix $\Sigma_0$. (We assume that the data have been centered in order to focus on the covariance matrix estimation problem alone, but it is important to point out that, in high dimensions, even estimating the mean vector is an intricate problem and much research has been conducted to address it.) The sample covariance matrix $S$, is a natural estimator of $\Sigma_0$ if $p$ is fixed and $n \to \infty$, but it can perform badly when $p$ is close to or larger than $n$, so some additional structural constraints are needed in order to facilitate estimation. In this paper, we study a particular type of such structural constraints. 

The main model assumption in our work here is that the population covariance matrix, $\Sigma_0 \in \mathbf{S}^{p}_{++}$, can be decomposed as 
$$\Sigma_0=L_{\Sigma_0}+D_{\Sigma_0},$$ 
in which $L_{\Sigma_0} \in \mathbf{S}^{p,r_0}_+$ is a row-rank matrix for some $r_0 \leq p$, and $D_{\Sigma_0} \in \mathbf{D}^p_{++}$ is a diagonal matrix. 

Such a decomposition is always possible as long as $r_0 \leq p$, but only for reasonably small $r_0$ is the assumed decomposition interesting and valuable for estimating large covariance matrices. Thus, for a particular matrix $\Sigma_0$, we define $r_0$ as the smallest among all attainable ranks of $L_{\Sigma_0}$ after the decomposition, i.e., $r_0=\rank(L^*)$ in which
\begin{eqnarray}
\label{eq:r0 def}
L^*&=&\underset{L}{\argmin}~\rank(L),\nonumber\\
\text{subject to}&& L+D=\Sigma_0,~
 L \in \mathbf{S}^{p}_+,~
 D \in \mathbf{D}^p_{++}.
\end{eqnarray}
As a solution of (\ref{eq:r0 def}), the matrix $L^*$ itself might not be unique, but the optimal value $r_0$ is.

How should one understand this model assumption conceptually?
As our first intuition, the assumption can be viewed as a generalization of the compound symmetry structure 
\[\left[
 \begin{matrix}
  a & b & \cdots & b \\
  b & a & \cdots & b \\
  \vdots  & \vdots  & \ddots & \vdots  \\
  b & b & \cdots & a
 \end{matrix}
 \right]
\]
with $a>b$, which was exploited by \cite{ourQDA} as a special structure to facilitate quadratic discriminant analysis in high dimensions. Notice that covariance matrices having the compound symmetry structure above can be decomposed into a rank-one matrix plus a scaled identity matrix, 
$$b1_p1_p^{\T}+(a-b)I_{p},$$ 
in which $1_p$ is a vector of ones and $I_{p}$ is the $p \times p$ identity matrix. Therefore, the compound symmetry structure can be seen as a special case of the ``low rank + diagonal'' decomposition.

The proposed decomposition also coincides with the factor analysis model and enjoys a nice interpretation. It is equivalent to assuming that the observed random vector $x$ depends on a potentially smaller number of latent factors, i.e., $x=Rz+\epsilon$, in which $z$ is some unobserved $r_0$-dimensional random vector from a normal distribution with mean $0$ and variance $I_{r_0}$, $R$ is an unobserved $p \times r_0$ loading matrix, and $\epsilon$ is a $p$-dimensional vector of independently distributed error terms with zero mean and finite variance, $\cov(\epsilon)=\Psi$. Under the given structure, it is straight-forward to see that $\cov(x)=RR^{\T}+\Psi$, in which $RR^{\T} \in \mathbf{S}^{p,r_0}_{+}$ is a low-rank matrix and $\Psi \in \mathbf{D}^p_{++}$ is a diagonal matrix. For our purpose, we are not interested in estimating the loading matrix or analyzing the latent factors; we merely exploit the special structure to help us estimate $\Sigma_0$. This purely ``utilitarian'' use of the factor model is also the reason why we can define $r_0$ simply as the smallest attainable rank in the ``low-rank + diagonal'' decomposition.

Finally, we can also think of the ``low-rank + diagonal'' assumption as an alternative to the popular sparsity assumption to facilitate the estimation of large covariance matrices. Numerous methods with lasso-type penalties assume a large number of zero off-diagonal entries in $\Sigma_0$; undoubtedly \emph{some} of these sparse structures can be represented as the sum of a low-rank matrix (i.e., with many empty rows and columns) and a diagonal matrix. The rank constraint is also somewhat analogous to the sparsity constraint. Specifically, the rank of $L_{\Sigma_0}$ is the number of its non-zero eigenvalues, so low-rank means its spectrum (i.e., set of eigenvalues) is sparse. Like the sparsity constraint, a rank constraint also reduces the total number of parameters to be estimated, as lower ranks of $L_{\Sigma_0}$ imply more linearly dependent columns and rows in $L_{\Sigma_0}$.

\section{PRECISION MATRIX ESTIMATION WITH FIXED RANK}
\label{sec:FR}
\subsection{The estimation method}
Our main model assumption can be equivalently imposed either on the covariance or on the corresponding precision matrix. Let $\Theta_0=\Sigma_0^{-1}$ be the corresponding precision matrix. To understand the structure of $\Theta_0$ when $\Sigma_0$ has the aforementioned ``low-rank + diagonal'' decomposition, we notice by a result of \cite{henderson1981deriving} that  
\begin{eqnarray}
\label{eq:L+D inv}
\left(L_{\Sigma_0}+D_{\Sigma_0}\right)^{-1}
&=&-D_{\Sigma_0}^{-1}
\left(I_p+L_{\Sigma_0}D_{\Sigma_0}^{-1}\right)^{-1}
L_{\Sigma_0}D_{\Sigma_0}^{-1}+D_{\Sigma_0}^{-1}\nonumber\\
&\triangleq&-L_0+D_0,
\end{eqnarray}
in which $L_0 \in \mathbf{S}^{p,r_0}_+$ and $D_0 \in \mathbf{D}^p_{++}$, because the product of several matrices has rank at most equal to the minimum rank of all the individual matrices in the product, and the inverse of a matrix in $\mathbf{D}^p_{++}$ is still in $\mathbf{D}^p_{++}$. 
Therefore, we see that the precision matrix $\Theta_0$ has an equivalent decomposition. 

With this in mind, we will henceforth concentrate on estimating the precision matrix rather than the covariance matrix. This is in line with various recent literatures on covariance matrix estimation; the precision matrix is also the more ``natural'' variable for maximizing the Gaussian-likelihood and the more ``direct'' quantity to use in many statistical procedures such as discriminant analysis. 
 
Other than the main ``low-rank + diagonal'' condition, our theoretical results will also require a ``bounded eigenvalue'' condition (see Condition \ref{(C.1)} below), which is purely technical but common in the literature. Thus, our entire set of conditions about the population covariance/precision matrix is as follows:
\begin{condition}
\label{(C.1)}
There exist constants $c_1, c_2 > 0$ such that 
$c_1\leq \lambda_{\min}(\Sigma_0) \leq \lambda_{\max}(\Sigma_0) \leq c_2$, or equivalently,
$c_2^{-1}\leq \lambda_{\min}(\Theta_0) \leq \lambda_{\max}(\Theta_0) \leq c_1^{-1}$.
\end{condition}

\begin{condition}
\label{(C.2)}
For some $r_0=o(p)$, the population covariance matrix $\Sigma_0 \in \mathbf{S}_{++}^p$ can be decomposed as $\Sigma_0=L_{\Sigma_0}+D_{\Sigma_0}$ where $L_{\Sigma_0} \in \mathbf{S}^{p,r_0}_+$  and $D_{\Sigma_0} \in \mathbf{D}_{++}^p$; or equivalently, the precision matrix $\Theta_0 \in \mathbf{S}^p_{++}$ can be decomposed as $\Theta_0=-L_0+D_0$, where $L_0 \in \mathbf{S}^{p,r_0}_+$ and $D_0 \in \mathbf{D}_{++}^p$. 
\end{condition}

In this section, we shall first consider a simple version of the problem, in which the rank of $L_0$ is pre-specified. We will consider the more general version of the problem later in Section~\ref{sec:rank-penalty}. One pragmatic reason for first considering a simple (and perhaps somewhat unrealistic) version of the problem is because our main result regarding the more general version and our computational algorithm for solving it are both based on results that we shall derive in this section for the simple version. 

For the simple version, a natural precision matrix estimator is 
\begin{eqnarray}
\label{eq:fr-mle}
(\hat{\Theta}_r,\hat{L}_r,\hat{D}_r)&=&\underset{\Theta}{\argmin}\{\trace(\Theta S)-\log|\Theta|\},\nonumber\\
\text{subject to}&& \Theta=-L+D,~
\Theta \in \mathbf{S}^{p}_{+},~
L \in \mathbf{S}^{p,r}_{+},~
D \in \mathbf{D}^p,
\end{eqnarray}
in which $r$ is a pre-specified constant. The objective function is the negative log-likelihood of the normal distribution, up to a constant. 
Let
\[
\mathbf{F}^r=\{\Theta \in \mathbf{S}_{++}^p \mid L \in \mathbf{S}^{p,r}_{+}, D \in \mathbf{D}_{++}^p \text{ and } \Theta = -L+D\}
\]
denote the search space of the optimization problem given in (\ref{eq:fr-mle}). 
In Sections~\ref{sec:FR-conservative} and \ref{sec:FR-aggressive} below, we will establish theoretical results to the following effects:
(i) if the pre-specified constant $r \geq r_0$, then the true precision matrix $\Theta_0 \in \mathbf{F}^r$, but if $r$ is much larger than $r_0$, the search space can be ``too large'' and solving (\ref{eq:fr-mle}) will become inefficient for estimating $\Theta_0$;
(ii) if the pre-specified constant $r < r_0$, then $\Theta_0 \notin \mathbf{F}^r$, and the gap between $\hat{\Theta}_r$ and $\Theta_0$ will depend on the distance between $\Theta_0$ and the search space $\mathbf{F}^r$. 

\noindent\textbf{Remark 1} \textit{In (\ref{eq:fr-mle}), it is unnecessary to explicitly restrict $\Theta$ or $D$ to be positive definite. The $-\log|\Theta|$ term in the objective function and the constraint $\Theta \in \mathbf{S}_+^p$ together will guarantee $\Theta \in \mathbf{S}_{++}^p$. In addition, as $\Theta = -L+D$ and $L \in \mathbf{S}_{+}^{p,r}$, we will also automatically have $D \in \mathbf{D}_{++}^p$, for $\Theta$ may not be in $\mathbf{S}_{++}^p$ otherwise.}

\noindent\textbf{Remark 2} \textit{The non-uniqueness of $\hat{L}_r$ and $\hat{D}_r$ is inconsequential for our purposes; our results and discussions below only depend on $\hat{\Theta}_r$ being a feasible minimizing solution to (\ref{eq:fr-mle}).}

\subsection{The conservative case: $r\geq r_0$}
\label{sec:FR-conservative}

To pre-specify the rank of $L_0$, denoted by $r$, it is generally advisable to err on the conservative side by choosing it to be large enough so that one can be more or less sure that $r \geq r_0$.
\begin{thm}
	Under Conditions \ref{(C.1)} and \ref{(C.2)}, if $r \geq r_0$ and $\hat{\Theta}_r$ is a solution of (\ref{eq:fr-mle}), then
\begin{align*}
\|\hat{\Theta}_r-\Theta_0\|_F=O_p\left\{\max(a_{n,p,r},b_{n,p})\right\},
\end{align*}
in which 
\[
a_{n,p,r}=r^{1/2}(p/n)^{1/2},\quad
b_{n,p\phantom{,r}}=\left\{(p\log p)/n\right\}^{1/2}.
\]
\label{thm:fr-Consistency}
\end{thm}

The true rank, $r_0$, may be fixed and finite, or it may diverge to infinity with $p$ and $n$. Since Theorem~\ref{thm:fr-Consistency} concerns the case of $r \geq r_0$, if $r_0 \to \infty$, then $r$ must necessarily also go to infinity. Hence, finite choices of $r \geq r_0$ are only possible if $r_0$ is also finite. If $r_0$ is finite and we choose a finite $r \geq r_0$, the consistency of $\hat{\Theta}_r$ is driven by $b_{n,p}$, whose order is greater than that of $a_{n,p,r}$, and the theorem basically suggests that choosing $r \geq r_0$ conservatively will not hurt estimation in any fundamental way. Otherwise if we must choose a diverging $r$, it becomes possible for the convergence rate to be driven by $a_{n,p,r}$, and the theorem basically implies that the estimator $\hat{\Theta}_r$ will be less efficient for larger, more conservative, choices of $r$.

\subsection{The aggressive case: $r<r_0$}
\label{sec:FR-aggressive}


What if one errs on the aggressive side by choosing $r$ to be too small so that $r < r_0$?
Let
\begin{eqnarray}
d_{r,r_0}&=& \underset{\Theta \in \mathbf{F}^r}{\min}||{\Theta-\Theta_0}||_F \nonumber
\end{eqnarray}
be the distance from $\Theta_0$ to the search space $\mathbf{F}^r$. When $r \geq r_0$, $d_{r,r_0}=0$. When $r<r_0$, the true precision matrix $\Theta_0$ is no longer in the search space $\mathbf{F}^r$, and $d_{r,r_0}>0$. Under such circumstances, it is still possible to achieve the same level of performance provided that $d_{r,r_0}$ is not too large.

\begin{thm}
Under Conditions \ref{(C.1)} and \ref{(C.2)}, if $r<r_0$, $d_{r,r_0}=O\{\max(a_{n,p,r_0},b_{n,p})\}$, and $\hat{\Theta}_r$ is a solution of (\ref{eq:fr-mle}), then
\begin{align*}
\|\hat{\Theta}_r-\Theta_0\|_F=O_p\left\{\max(a_{n,p,r_0},b_{n,p})\right\},
\end{align*}
in which 
\[
a_{n,p,r_0}=r_0^{1/2}(p/n)^{1/2},\quad
b_{n,p\phantom{,r_0}}=\left\{(p\log p)/n\right\}^{1/2}. 
\]
\label{thm:r<r0 consistency}
\end{thm}

While the proof itself is given in the appendices, the main reason why Theorem~\ref{thm:r<r0 consistency} holds is as follows. Let $\Theta_r \in \mathbf{F}^r$ be the matrix closest to $\Theta_{0}$ such that $\|\Theta_r-\Theta_0\|_F=d_{r,r_0}$. It can be shown that $\hat{\Theta}_r$, as the solution to maximizing the likelihood function in the search space $\mathbf{F}^r$, will be close to $\Theta_r$. So, if $d_{r,r_0}$ is small, $\hat{\Theta}_r$ will also be reasonably close to $\Theta_0$. More importantly, the condition $d_{r,r_0}=O\{\max(a_{n,p,r_0},b_{n,p})\}$ requires the distance $d_{r,r_0}$ to be of order $\max(a_{n,p,r_0},b_{n,p})$, which, by Theorem~\ref{thm:fr-Consistency}, is also the order of the estimation error when the rank $r$ is correctly set to be $r_0$. As a result, the error caused by $\Theta_0$ being away from $\mathbf{F}^r$ is relatively small and does not increase the order of the estimation error. 


However, by definition $d_{r,r_0}$ is also a lower bound for the estimation error,
\begin{eqnarray}
\|\hat{\Theta}_r-\Theta_0\|_F& \geq & d_{r,r_0}, \nonumber
\end{eqnarray}
which means, not surprisingly, that $\hat{\Theta}_r$ will cease to be a consistent estimator of $\Theta_0$ if $d_{r,r_0}$ is large. 

\subsection{Discussion}

To summarize what we have presented so far, although the optimization problem (\ref{eq:fr-mle}) is straight-forward and easy to implement (see Section~\ref{sec:algorithm}), it is generally not possible to specify $r$ accurately. An inaccurate choice of $r$ can be harmful in two ways: (1) A conservative choice of $r>r_0$ 
leads to slower convergence and less estimation efficiency. (2) An aggressive choice of $r<r_0$ can ruin the consistency of $\hat{\Theta}_r$, because it can enlarge the distance between $\Theta_0$ and the search space $\mathbf{F}^r$. 

In the next section, we introduce a rank penalty to circumvent these problems. However, our main result below (Theorem~\ref{thm:FR-penalty-consistency}) as well as the main computational algorithm (Section~\ref{sec:algorithm}) are both heavily based on the results (Theorems~\ref{thm:fr-Consistency} and \ref{thm:r<r0 consistency}) that we have obtained so far in this section.

\section{PRECISION MATRIX ESTIMATION WITH RANK PENALTY}
\label{sec:rank-penalty}
\subsection{The estimation method}
One way to avoid having to specify the rank of the low-rank component $L$ is by adding a penalty on the rank of $L$ to the objective function in (\ref{eq:fr-mle}).
That is, instead of (\ref{eq:fr-mle}), we can solve the following optimization problem:
\begin{eqnarray}
\label{eq:rank-MLE}
(\hat{\Theta},\hat{L},\hat{D})&=&\underset{\Theta,L,D}{\argmin}\left[\trace(\Theta S)-\log|\Theta|+\tau\{\rank(L)\}\right],\nonumber\\
\text{subject to}&& -L+D=\Theta,~
\Theta \in \mathbf{S}^{p}_{+},~
L \in \mathbf{S}^{p}_{+},~
D \in \mathbf{D}^p,
\end{eqnarray}
where $\tau(\cdot)$ is a monotonically increasing penalty function. 

In the literature, it is popular to impose rank restrictions on a matrix by penalizing its nuclear norm. There are some advantages to directly penalizing its rank. Let $\hat{\Theta}_r$ denote the solution to (\ref{eq:fr-mle}). Clearly, if we fix $\rank(L)=r$ in (\ref{eq:rank-MLE}), its solution becomes $\hat{\Theta}=\hat{\Theta}_r$. This means $\hat{\Theta}$ can only be one of $\{\hat{\Theta}_r \mid r=1,\dots,p\}$, which will have a direct implication on how (\ref{eq:rank-MLE}) can be solved in practice. In particular, we shall see in Section~\ref{sec:algorithm} below that, for fixed $r$, $\hat{\Theta}_r$ can be obtained by a relatively efficient blockwise coordinate descent algorithm, in which the update of $L$ given $D$ can be achieved analytically, and the update of $D$ given $L$ is a relatively cheap log-determinant semi-definite program.  

In this section, however, we shall concentrate on the key question of how to choose the penalty function $\tau(\cdot)$ in order to ensure that $\hat{\Theta}$ is a good estimator of $\Theta_0$. Our answer is that it must satisfy the following two conditions:
\begin{condition}
\label{(C.3)}
If $r<r_0$ and $d_{r,r_0}/\max(a_{n,p,r_0},b_{n,p}) \to \infty$, then $|\tau(r)-\tau(r_0)|/d_{r,r_0}^2 \to 0$.
\end{condition}
\begin{condition}
\label{(C.4)}
If $r>r_0$ and $r/\max(r_0,\log p) \to \infty$, then $a_{n,p,r}^2/|\tau(r)-\tau(r_0)| \to 0$.
\end{condition}
These conditions are quite technical, and readers will find a concrete example of $\tau(\cdot)$, to be provided later in Section~\ref{sec:rank-penalty-example}, much easier to grasp. Our main result is that, with a penalty function that satisfies Conditions \ref{(C.3)} and \ref{(C.4)}, the solution of (\ref{eq:rank-MLE}) will be a good estimator of $\Theta_0$.
\begin{thm}
Under Conditions \ref{(C.1)}, \ref{(C.2)}, \ref{(C.3)} and \ref{(C.4)}, if $\hat{\Theta}$ is a solution of (\ref{eq:rank-MLE}), then
\begin{align*}
\|\hat{\Theta}-\Theta_0\|_F=O_p\left\{\max(a_{n,p,r_0},b_{n,p})\right\},
\end{align*}
in which 
\[
a_{n,p,r_0}=r_0^{1/2}(p/n)^{1/2},\quad
b_{n,p\phantom{,r_0}}=\left\{(p\log p)/n\right\}^{1/2}.
\]
\label{thm:FR-penalty-consistency}
\end{thm}
Comparing the conclusion of Theorem~\ref{thm:FR-penalty-consistency} with that of Theorem~\ref{thm:fr-Consistency}, we can see that the convergence rates of the two methods, whether using a penalty on $\rank(L)$ or a pre-specified rank for $L$, are similar. The only difference is that the convergence rate of the former depends on the true rank $r_0$, as long as the penalty function $\tau(\cdot)$ is chosen appropriately, while the convergence rate of the latter depends on the presumed rank $r$. 


\subsection{Technical conditions on the penalty function}
\label{sec:rank-penalty-general}

To understand Conditions \ref{(C.3)} and \ref{(C.4)}, and how they are essential to Theorem~\ref{thm:FR-penalty-consistency}, let us partition the set $\{r \mid r \neq r_0\}$ into four disjoint pieces:
\begin{eqnarray*}
\mathbf{A}_1 &=& \{r \mid r<r_0, d_{r,r_0}/\max(a_{n,p,r_0},b_{n,p}) \to \infty\},\\
\mathbf{A}_2 &=& \{r \mid r<r_0, d_{r,r_0}=O[\max(a_{n,p,r_0},b_{n,p})]\},\\
\mathbf{A}_3 &=& \{r \mid r>r_0, r=O[\max(r_0,\log p)]\},\\
\mathbf{A}_4 &=& \{r \mid r>r_0, r/\max(r_0,\log p) \to \infty\}.
\end{eqnarray*}
Notice that, by definition, for any $r_i \in \mathbf{A}_i$ ($i=1,2,3,4$), we have $r_1<r_2<r_0<r_3<r_4$. 

Together, Theorem~\ref{thm:fr-Consistency} and Theorem~\ref{thm:r<r0 consistency} have already established the convergence rate of $\hat{\Theta}_r$ to be $\max(a_{n,p,r_0},b_{n,p})$ for $r \in \mathbf{A}_2 \cup \mathbf{A}_3 \cup \{r_0\}$. 
A penalty function that satisfies Conditions \ref{(C.3)} and \ref{(C.4)} will ensure that 
the solution to (\ref{eq:rank-MLE}) cannot be in the set $\{\hat{\Theta}_r \mid r \in \mathbf{A}_1 \cup \mathbf{A}_4\}$. 

Specifically, as $\|\hat{\Theta}_r-\Theta_0\|_F \geq d_{r,r_0}$, any $\hat{\Theta} \in \{\hat{\Theta}_r \mid r \in \mathbf{A}_1\}$ cannot achieve the convergence rate given in Theorem~\ref{thm:FR-penalty-consistency}, but Condition \ref{(C.3)} ensures that such a $\hat{\Theta}$ will not be chosen by (\ref{eq:rank-MLE}) because, if $r \in \mathbf{A}_1$, we have \[\trace(\hat{\Theta}_rS)-\log{|\hat{\Theta}_r|} \geq \trace(\hat{\Theta}_{r_0}S)-\log{|\hat\Theta_{r_0}|},\] and \[\tau{(r)}<\tau{(r_0)}.\] 
The first inequality encourages the optimization problem (\ref{eq:rank-MLE}) to favor a solution with $\rank(L)=r_0$ while the second inequality encourages it to favor one with a smaller rank, $r$. Condition \ref{(C.3)} will ensure that $\tau(r_0)-\tau(r)$ is relatively small so that the influence from the penalty function (the second inequality above) will remain relatively weak.
Likewise, by Theorem~\ref{thm:fr-Consistency}, any $\hat{\Theta} \in \{\hat{\Theta}_r \mid r\in \mathbf{A}_4\}$ cannot achieve the convergence rate given in Theorem~\ref{thm:FR-penalty-consistency}, either, but Condition \ref{(C.4)} will ensure that, for $r\in \mathbf{A}_4$, $\tau{(r)}-\tau{(r_0)}$ is sufficiently large so that the influence from the penalty function is strong enough to outweigh the fact that the first inequality above has now switched direction for $r\in\mathbf{A}_4$.


\subsection{A concrete example}
\label{sec:rank-penalty-example}

At this point, it will help greatly to see a concrete example of penalty functions that satisfy Conditions \ref{(C.3)} and \ref{(C.4)}.
Given $n$ observations from a $p$-dimensional multivariate Gaussian model with a rank-$r$ covariance matrix, where $r \leq p$, \cite{akaike1987factor} derived that the Akaike information criterion (AIC) is 
\begin{eqnarray}
\label{eq:AIC}
\AIC(r) &=& \frac{1}{n} \left[ (-2)\sum_{i=1}^n \ell(x_i)+\left\{2p(r+1)-r(r-1)\right\} \right],
\end{eqnarray}
where $\ell(x)$ denotes the log-density function.
A penalty function that satisfies both Conditions \ref{(C.3)} and \ref{(C.4)} is
\begin{eqnarray}
\label{eq:penalty}
\tau(r)&=&\delta_{n,p}\left\{2p(r+1)-r(r-1)\right\}/n,
\end{eqnarray}
in which 
\begin{eqnarray}
\label{eq:delta np infinity}
\delta_{n,p} &\to& \infty,
\end{eqnarray}
and
\begin{eqnarray}
\label{eq:delta np upper bound}
\delta_{n,p} &=& o\left\{d_{r,r_0}^2 n/(r_0p)\right\} \quad\mbox{for all}\quad r \in \mathbf{A}_1.
\end{eqnarray}
Therefore, we can see that (\ref{eq:penalty}) is essentially a scaled version of the AIC penalty. The condition (\ref{eq:delta np infinity}) on the scaling factor $\delta_{n,p}$ means that the penalty (\ref{eq:penalty}) is slightly larger than the AIC penalty asymptotically. 

For all $r \in \mathbf{A}_1$, $d_{r,r_0}^2/(r_0p/n) \to \infty$ by definition, so (\ref{eq:delta np upper bound}) does not contradict with (\ref{eq:delta np infinity}); it is also equivalent to 
\begin{eqnarray*}
\delta_{n,p} &=& o\left[\underset{r \in \mathbf{A}_1}{\min}\{d_{r,r_0}^2n
/(r_0p)\}\right].
\end{eqnarray*}
To verify that (\ref{eq:penalty}) satisfies Conditions \ref{(C.3)} and \ref{(C.4)}, notice that  
\begin{eqnarray*}
\tau(r)-\tau(r_0)&=& \delta_{n,p}(r-r_0)\{2p-(r+r_0-1)\}/n.
\end{eqnarray*}
On the one hand, any given $r<r_0$ such that $d_{r,r_0}/\max(a_{n,p,r_0},b_{n,p}) \to \infty$ is in the set $ \mathbf{A}_1$ and 
\begin{eqnarray*}
|\tau(r)-\tau(r_0)|/d_{r,r_0}^2&=&\delta_{n,p}(r-r_0)\{2p-(r+r_0-1)\}/(d_{r,r_0}^2n)\\
&=& o\left[(r-r_0)\{2p-(r+r_0-1)\}/(r_0p)\right]\\
&=& o(1),
\end{eqnarray*} 
so Condition \ref{(C.3)} is satisfied. On the other hand, any given $r > r_0$ such that $r/\max(r_0,\log p) \to \infty$ is in the set $\mathbf{A}_4$ and
\begin{eqnarray*}
a_{n,p,r}^2/|\tau(r)-\tau(r_0)|&=&rp/[\delta_{n,p}(r-r_0)\{2p-(r+r_0-1)\}]\\
&=& o(1),
\end{eqnarray*}
so Condition \ref{(C.4)} is satisfied.
\subsection{Discussion}

The convergence rate given by Theorem~\ref{thm:FR-penalty-consistency} applies both to finite $r_0$ and to $r_0$ that may diverge to infinity with $p$ and $n$. If $r_0$ is fixed and finite, the consistency of $\hat{\Theta}$ is driven by $b_{n,p}=[(p\log p)/n]^{1/2}$, whose order is greater than that of $a_{n,p,r_0}$; otherwise, it is possible for the convergence rate to be driven by $a_{n,p,r_0}=r_0^{1/2}(p/n)^{1/2}$ --- e.g., if $r_0$ goes to infinity faster than does $\log p$.

One can better assess our convergence rate here in the Frobenius norm by comparing it with the convergence rate of the ``sparse precision matrix estimator'' given by \cite{rothman2008sparse}. Their convergence rate in the Frobenius norm is $\{(p+s)(\log p)/n\}^{1/2}$, in which $s$ is the number of nonzero off-diagonal entries in the population precision matrix. For fixed $s$, their rate becomes $\{(p\log p)/n\}^{1/2}$ and is the same as our rate ($b_{n,p}$) for fixed $r_0$. 

That these convergence rates are of a comparable order provides another argument that the low-rank assumption can be regarded as an analogue of the sparsity assumption for estimating high-dimensional covariance/precision matrices, except that it encourages a slightly different matrix structure.

\section{A BLOCKWISE COORDINATE DESCENT ALGORITHM}
\label{sec:algorithm}

We now describe a computational algorithm for solving the optimization problem (\ref{eq:rank-MLE}). As we have pointed out in Section~\ref{sec:rank-penalty}, the solution to (\ref{eq:rank-MLE}) can only be one of $\{\hat{\Theta}_r: r=1,2,...,p\}$. In principle, this means we can simply solve (\ref{eq:fr-mle}) for all $r\in\{1,...,p\}$ and choose the one that minimizes the objective function (\ref{eq:rank-MLE}). In practice, it is usually sufficient, and not impractical, to do this only on a subset of $\{1,2,...,p\}$, say $\mathbb{Z}_r$. 

That is, we first obtain a series of fixed-rank estimators, $\hat{\Theta}_r$, by solving (\ref{eq:fr-mle}) for each $r \in \mathbb{Z}_r$. Then, we use the penalty function (\ref{eq:penalty}), given in Section \ref{sec:rank-penalty-example}, and evaluate the objective function (\ref{eq:rank-MLE}) at each $\{\hat{\Theta}_r \mid r \in \mathbb{Z}_r\}$, and the one that minimizes the objective function (\ref{eq:rank-MLE}) is taken as the solution, $\hat{\Theta}$. As we do not have an explicit expression for $\delta_{n,p}$, it is treated in practice as a tuning parameter 
and selected by minimizing the objective function on a separate, validation data set. 

For each $r \in \mathbb{Z}_r$, $\hat{\Theta}_r$ is obtained by solving the fixed-rank optimization problem (\ref{eq:fr-mle}) with a blockwise coordinate descent algorithm, which iteratively updates $L$ and $D$ (see Algorithm~\ref{algo}). For fixed $D$, we can actually solve for $L$ analytically; this provides an enormous amount of computational saving. The validity of Step 2, the analytic update of $L$ given $D$, is established by Lemma 4 in the appendices. For fixed $L$, we solve a log-determinant semi-definite program over $D$, e.g., using the SDPT3 solver \citep{tutuncu2003solving} available as part of the YALMIP toolbox \citep{lofberg2004yalmip} in Matlab; the fact that $D$ is diagonal means the log-determinant semi-definite program here is one of the cheapest kinds to solve. 

To initialize the blockwise coordinate descent algorithm for each $r \in \mathbb{Z}_r$, we suggest arranging all $r \in \mathbb{Z}_r$ in ascending order and solving for each $\hat{\Theta}_r$ sequentially, using the last solution as a ``warm start'' for finding the next solution. To be more specific, for $r^{(1)} < r^{(2)} < ... \in \mathbb{Z}_r$, we suggest using the diagonal component of $\hat{\Theta}_{r^{(k-1)}}$, namely $\hat{D}_{r^{(k-1)}}$, as the initial point ($D^{(0)}$ in Algorithm~\ref{algo}, Step 1) for obtaining $\hat{\Theta}_{r^{(k)}}$. To initialize the algorithm for the very first $\hat{\Theta}_{r^{(1)}}$, we suggest using the solution of (\ref{eq:fr-mle}) corresponding to $r=0$; taking $r=0$ means there is no low-rank component, so we have an analytical solution, $D^{(0)}=\hat{D}_{0}=\diag\{s_{11}^{-1},\ldots,s_{pp}^{-1}\}$, where $s_{jj}$ is the $j$th diagonal element of the sample covariance matrix $S$. Our experience from running many numerical experiments shows that obtaining $\hat{\Theta}_r$ in such a sequential manner is much more efficient than obtaining each $\hat{\Theta}_r$ independently with random ``cold start'' initialization. 

\noindent\textbf{Remark 3} \textit{We think Lemma 4, the analytic update of $L$ given $D$, is a useful piece of contribution on its own. It can be used to obtain other ``low-rank + something'' type of decompositions of precision matrices, as the low-rank step (Step 2 of the algorithm) does not depend on $D$ being diagonal. For example, one can assume that $D$ is a sparse matrix and the coordinate descent algorithm (Algorithm~\ref{algo}) can still be applied, as long as one modifies Step 3 to include a sparsity penalty such as $\|D\|_1=\sum_{i,j} |D_{ij}|$, although we generally will expect the resulting Step 3 to become more computationally expensive than it is when $D$ is diagonal.}

\begin{algorithm}
Blockwise coordinate descent algorithm for solving (\ref{eq:fr-mle}) for each $r \in \mathbb{Z}_r$.
\begin{tabbing}
  \enspace Step 1: Initialize $D^{(0)}$.\\
  \enspace Step 2: Fix $D^{(i)}$ and update $L^{(i+1)}$ analytically. \\
   \qquad -- Obtain the eigen-decomposition of $(D^{(i)})^{1/2}S(D^{(i)})^{1/2}$.\\          
   \qquad -- Let $w_1^{(i)},\ldots,w_r^{(i)}$ denote the $r$ largest eigenvalues. \\
   \qquad -- Let $u_1^{(i)},\ldots,u_r^{(i)}$ denote the corresponding eigenvectors.\\
   \qquad -- Set $U^{(i+1)}=[\begin{array}{ccc}u_1^{(i)} & \ldots & u_r^{(i)}\end{array}]$.\\
   \qquad -- Set $V^{(i+1)}=\diag\{1-1/\max{(w_1^{(i)},1)},\ldots,1-1/\max{(w_r^{(i)},1)}\}$, \\
   \qquad -- Set $L^{(i+1)}=(D^{(i)})^{1/2}U^{(i+1)} V^{(i+1)} (U^{(i+1)})^{\T}(D^{(i)})^{1/2}$.\\
  \enspace Step 3: Fix $L^{(i)}$ and update $D^{(i+1)}$ by solving a log-determinant semi-definite program.\\
   \qquad -- Minimize $\trace\{(D-L^{(i)})S\}-\log|D-L^{(i)}|$ over $D$.\\
\enspace  Step 4: Repeat Step 2 and 3 until 
$\trace\{(D^{(i)}-L^{(i)})S\}-\log|D^{(i)}-L^{(i)}|$ converges. 
\end{tabbing}
\label{algo}
\end{algorithm}

\section{SIMULATION}
\label{sec:numerical}
\subsection{Simulation settings}
In this section, we compare four different estimators of the covariance/precision matrix: the sample covariance matrix ($S$); a simple diagonal estimator ($D_S$), which keeps only the diagonal elements of $S$ and sets all off-diagonal elements to zero; the graphical lasso (Glasso) by \cite{friedman2008sparse}; and our method (LD). The graphical lasso is implemented with the \verb!R! package \verb!glasso!.

Using a training sample size of $n=100$, we generated data from $p$-dimensional ($p=50, 100, 200$) normal distributions with mean $0$ and the following five population covariance matrices:
\begin{example}
\label{ex.1}
The matrix $\Sigma_1$ is compound symmetric, $\Sigma_1=(0.2)1_p1_p^{\T}+(0.8)I_p$.
\end{example}
\begin{example}
\label{ex.2}
The matrix $\Sigma_2$ is ``low-rank + diagonal'', $\Sigma_2=I_p+RR^{\T}$, 
where $R \in \mathbf{R}^{p\times5}$ and all of its elements are independently sampled from the Uniform$(0,1)$ distribution.
\end{example}
\begin{example}
\label{ex.3}
The matrix $\Sigma_3$ is block diagonal, consisting of $5$ identical blocks $B=(0.2)1_q1_q^{\T}+(0.8)I_q$, where $q=p/5$.
\end{example}
\begin{example}
\label{ex.4}
The matrix $\Sigma_4$ is almost ``low-rank + diagonal'' but with some perturbations. First, a ``low-rank + diagonal'' matrix is created, $B_0=I_p+RR^{\T}$, where $R \in \mathbf{R}^{p \times 3}$ and all of its elements are independently sampled with probability $0.8$ from the Uniform$(0,1)$ distribution and set to $0$ otherwise. Next, a perturbation matrix $B_1 \in \mathbf{R}^{p \times p}$ is created, whose elements are independently sampled with probability $0.05$ from the Uniform$(-0.05,0.05)$ distribution and set to $0$ otherwise. Then, the perturbation matrix $B_1$ is symmetrized before being combined with $B_0$ to obtain $B=\left\{B_0^{-1}+(B_1+B_1^{\T})/2\right\}^{-1}$. Finally, we let $\Sigma_4=B+\delta I_p$, with $\delta = |\min(\lambda_{\min}(B),0)|+0.05$, to ensure it is positive definite.
\end{example}
\begin{example}
\label{ex.5}
The matrix $\Sigma_5$ is designed to have a sparse inverse. First, a baseline matrix $B_0 \in \mathbf{R}^{p \times p}$ is created where all of its elements are set to $0.5$ with probability $0.5$ and $0$ otherwise. Then, it is symmetrized and made positive definite before being inverted: $B=B_0+B_0^{\T}$, $\delta = |\min(\lambda_{\min}(B),0)|+0.05$, and $\Sigma_5=\left(B+\delta I_p\right)^{-1}$.
\end{example}
Each population covariance matrix in the first three examples can be decomposed into a low-rank plus a diagonal matrix. Let the decomposition be $\Sigma_k=L_{\Sigma_k}+D_{\Sigma_k}$ for $k=1,2,3$; then, $L_{\Sigma_1} \in \mathbf{S}^{p,1}_{+}$ and $L_{\Sigma_2}, L_{\Sigma_3} \in \mathbf{S}^{p,5}_{+}$. Example \ref{ex.4} is used to test the robustness of our method; starting from a ``low-rank + diagonal'' matrix, we randomly perturbed approximately $10\%$ of the elements in the corresponding precision matrix. Example \ref{ex.5} is used to illustrate the performance of our method in a situation that is ideal to the graphical lasso, where the
corresponding precision matrix is sparse.

Tuning parameters are selected by minimizing the negative log-likelihood function on a separate validation data set of size $100$. For the graphical lasso, the tuning parameter was selected from $\{0.01,0.03,0.05,0.07,0.09,$ $0.11,0.15,0.20\}$. For our method, we used $\mathbb{Z}_r=\{1,3,5,7,9\}$, and the tuning parameter $\delta_{n,p}$ was selected from $\{0.6,0.8,1.0,1.2,1.4\}$. Recall from Section~\ref{sec:algorithm} that only the size of $\mathbb{Z}_r$ affects our computational time, not the number of tuning parameters we evaluate. 

\subsection{Estimation accuracy}

As \cite{rothman2008sparse}, we evaluated the estimation accuracy with the Kullback--Leibler loss, 
\begin{align}
L_{KL}\left(\hat{\Theta},\Theta_0\right)=
\trace\left(\Theta_0^{-1}\hat{\Theta}\right)-\log|\Theta_0^{-1}\hat{\Theta}|-p.
\label{eq:KL loss}
\end{align}
When $\hat{\Theta}=\Theta_0$, the true precision matrix, the loss achieves its minimum of zero. For the graphical lasso and our method, the estimated precision matrix $\hat{\Theta}$ could be directly plugged into the loss function (\ref{eq:KL loss}); for $S$ and $D_S$, the estimated covariance matrix needed to be inverted first. Thus, we could not evaluate the loss for $S$ when $p=100$ and $p=200$, because it was non-invertible.

Table \ref{tab:simu-result} reports the average Kullback--Leibler loss over $100$ replications and its standard error. Not surprisingly, the sample covariance matrix $S$ was the worst estimator; the diagonal estimator $D_S$ was better in most cases, but not as good as the other two methods. In the first four examples, our method outperformed the graphical lasso. In Example \ref{ex.5}, an ideal case for the graphical lasso in which the population precision matrix was sparse, our method performed slightly worse than, but still remained largely competitive against, the graphical lasso.

\begin{table}
\def~{\hphantom{0}}
\caption{Average (standard error) of Kullback--Leibler loss over $100$ replications.
 \label{tab:simu-result}}
 \begin{center}
   \begin{tabular}{llcccc}
   \hline
          &       & $S$     & $D_S$     & Glasso & LD \\    
		  \hline\\
    Example 1& $p = 50$ & 37.59 (0.311) & 9.058 (0.011) & 2.618 (0.016) & 0.980 (0.019) \\
          & $p = 100$ & NA    & 20.09 (0.017) & 5.496 (0.029) & 1.983 (0.026) \\
          & $p = 200$ & NA    & 42.73 (0.024) & 11.39 (0.050) & 3.893 (0.040) \\         [5pt]
    Example 2 & $p = 50$ & 37.44 (0.331) & 36.80 (0.019) & 4.148 (0.024) & 2.751 (0.030) \\
          & $p = 100$ & NA    & 80.70 (0.043) & 9.469( 0.044) & 5.708 (0.043) \\
          & $p = 200$ & NA    & 170.0 (0.071) & 20.38 (0.082) & 11.85 (0.060) \\         [5pt]
    Example 3 & $p = 50$ & 37.67 (0.341) & 5.417 (0.011) & 3.080 (0.026) & 3.247 (0.038) \\
          & $p = 100$ & NA    & 14.40 (0.016) & 7.643 (0.038) & 6.103 (0.046) \\
          & $p = 200$ & NA    & 34.72 (0.022) & 16.48 (0.074) & 12.00 (0.076) \\[5pt]         
    Example 4& $p = 50$ & 37.52 (0.316) & 26.21 (0.017) & 3.522 (0.023) & 2.028 (0.022) \\
          & $p = 100$ & NA    & 33.00 (0.019) & 7.534 (0.040) & 3.917 (0.036) \\
          & $p = 200$ & NA    & 136.4 (0.062) & 16.35 (0.066) & 9.044 (0.057) \\[5pt]         
    Example 5 & $p = 50$ & 37.57 (0.312) & 42.80 (0.020) & 8.267 (0.034) & 9.949 (0.046) \\
          & $p = 100$ & NA    & 78.15 (0.028) & 22.03 (0.047) & 24.13 (0.073) \\
          & $p = 200$ & NA    & 180.1 (0.035) & 59.89 (0.096) & 61.08 (0.123)  \\
		  \hline  
    \end{tabular}
	\end{center}
\end{table}

\subsection{Rank recovery}
We also investigated how well $r_0$ was recovered by comparing the $10$ largest eigenvalues of $\hat{L}$ with those of $L_0$, the low-rank component of the population precision matrix. According to (\ref{eq:L+D inv}), $L_0$ can be derived as 
\[
L_0=
D_{\Sigma_0}^{-1}
\left(I+L_{\Sigma_0}D_{\Sigma_0}^{-1}\right)^{-1}
L_{\Sigma_0}D_{\Sigma_0}^{-1}.
\]
For Examples~\ref{ex.1}--\ref{ex.3}, the components $L_{\Sigma_0}$ and $D_{\Sigma_0}$ could be obtained directly from the set-up. For Example \ref{ex.4}, because of the perturbation, the components were only approximate: $L_{\Sigma_4} \approx RR^{\T}$ where $R \in \mathbf{R}^{p \times 3}$, and $D_{\Sigma_4} \approx I_p$. We skip Example~\ref{ex.5} here because the true covariance/precision matrix does not have a corresponding low-rank component. 

As the results were similar for different values of $p$, we only present here those for $p = 100$. In Fig.~\ref{fig:Pre}, the $10$ largest eigenvalues of $L_0$ and of $\hat{L}$ are plotted. For $\hat{L}$, the bigger dots in the middle are the averages over $100$ replications; the smaller dots above and below are the values, $(\text{average}) \pm (1.96)(\text{standard error})$.  We can see that on average our method successfully identified the nonzero eigenvalues, or the rank, of $L_0$. 

\begin{figure}[htbp]
\centering
\includegraphics[width=\textwidth]{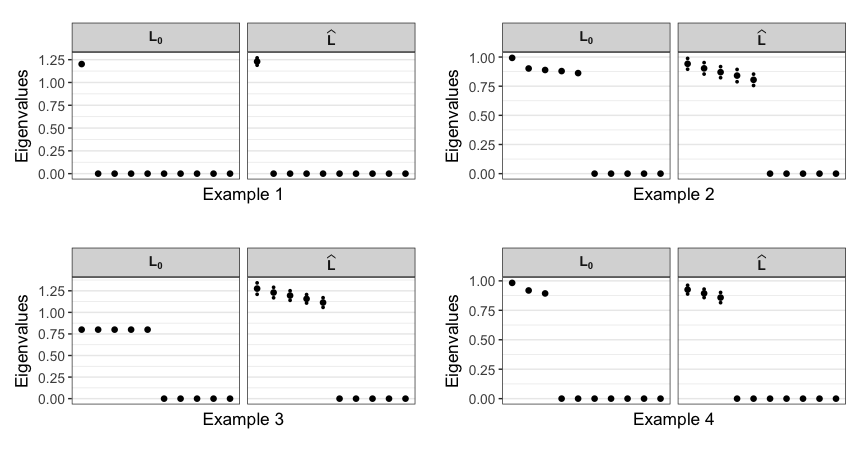}
\caption{Comparison of the $10$ largest eigenvalues of $L_0$ and those of $\hat{L}$ [$(\text{average}) \pm (1.96)(\text{standard error})$].}
\label{fig:Pre}
\end{figure}

\section{REAL DATA ANALYSIS}
\label{sec:realdata} 
To showcase a real application of our method to high-dimensional covariance/precision matrix estimation, we discuss the classic Markowitz portfolio selection problem \citep{markowitz1952portfolio}. In this problem, we have the opportunity to invest in $p$ assets, and the aim is to determine how much to invest in each asset so that a certain level of expected return is achieved while the overall risk is minimized. To be more specific, let $\mu$ be the mean returns of $p$ assets and $\Sigma$, their covariance matrix. Let $1_p$ be the $p$-dimensional vector $(1,1,...,1)^{\T}$. Then, the Markowitz problem is formulated as
\begin{align}
\label{eq:Markowitz}
\hat{w}=\argmin w^{\T}\Sigma w
\text{~~~subject to~~~}
w^{\T}\mu=\mu_{0}, w^{\T} 1_p=1,
\end{align}
in which $w$ is a vector of asset weights, $\mu_0$ is the desired level of expected return, and $w^{\T}\Sigma w$ is the variance of the portfolio, which quantifies the investment risk.

In practice, $\mu$ and $\Sigma$ can be estimated respectively by the sample mean and the sample covariance matrix before the optimization problem (\ref{eq:Markowitz}) is solved, provided that the sample size $n$ is much larger than the dimension $p$; in high dimensions, however, solving (\ref{eq:Markowitz}) with the sample covariance matrix often leads to undesirable \emph{risk underestimation} \citep{el2010high}. Instead, different estimators of the covariance matrix can be used, such as those we have studied in the previous section: namely, the diagonal estimator ($D_S$), the graphical lasso (Glasso), and our method (LD).

To compare these different covariance matrix estimators for solving the Markowitz problem, we used monthly stock return data of companies in the S\&P100 index from January 1990 to December 2007, as did \cite{xue2012positive}. This dataset contains $p=67$ companies that 
remained in the S\&P100 throughout this entire period; for each stock, there are $12\times(2007-1990+1)=216$ monthly returns. 

For each month starting in January 1996, we first constructed a portfolio by solving the Markowitz problem using an estimated $\mu$ and $\Sigma$ from the preceding $n=72$ monthly returns, and a target return of $\mu_0=1.3\,\%$. The performance of the resulting portfolio was then measured by its return in that month. For any given estimator of $\Sigma$, a total of $12\times(2007-1996+1)=144$ portfolios were constructed and evaluated in this manner.

We used three-fold cross-validation to choose the tuning parameters for both the graphical lasso and our method. Each time, portfolios were constructed based on two-thirds of the training data (48 months), and the tuning parameter that maximized the average return on the remaining one-third of the training data (24 months) was selected. For the graphical lasso, the tuning parameter was selected from $\{0.2,0.4,\ldots,3.0\}$. For our method, we chose from the same set of tuning parameters, and the candidate ranks we considered, $\mathbb{Z}_r$, consisted of all even numbers between $2$ and $28$.

Table \ref{tab:financial-data-result} shows the results. Again, the sample covariance matrix was noticeably outperformed by all of the other three methods. Our method (LD) was better than $D_S$ in terms of both the average return and the overall volatility (standard error). Comparing with the graphical lasso, although our average return was slightly lower, our portfolio  had much lower volatility, and hence a higher Sharpe ratio, a popular measure of overall portfolio performance in finance defined as $[\text{mean}(x-x_b)]/[\text{stdev}(x-x_b)]$, where $x$ is the portfolio's and $x_b$ is the risk-free rate of return. For this demonstration here, we simply took $x_b=0$ to be constant. 

\begin{table}
\def~{\hphantom{0}}
  \caption{Average, standard error, and Sharpe ratio of monthly portfolio returns, January 1996 to December 2007. All numbers are expressed in \%.
  \label {tab:financial-data-result}} 
  \begin{center}
   \begin{tabular}{lcccc}
  \hline
 & $S$ & $D_S$ & Glasso & LD\\\hline
Average & 0.70 & 1.32 & 1.42 & 1.41 \\
Standard Error & 13.2 & 5.08 & 5.13 & 4.73\\
Sharpe ratio & 5.30 & 26.0 & 27.7& 29.8\\
\hline
    \end{tabular}
\end{center}
\end{table}

\section{CONCLUSION} 
We have proposed a high-dimensional covariance/precision matrix estimation method that decomposes the covariance/precision matrix into a low-rank plus a diagonal matrix. This structural assumption can be understood as being driven by a factor model and as an alternative to the popular sparsity assumption to facilitate estimation in high-dimensional problems. We estimate the precision instead of the covariance matrix because the resulting negative log-likelihood function is convex and because the precision matrix can be directly applied in many statistical procedures. 

Starting with a fixed-rank estimator, we have shown how it can be used to provide a more general estimator by maximizing a penalized likelihood criterion. Unlike \cite{taeb2016interpreting}, who used a nuclear-norm penalty to constrain the rank, we impose a penalty directly on the matrix rank itself. 

The theoretical conditions for a valid penalty function have been studied in general, and a specific example, which is related to the Akaike information criterion, has been discussed and tested. Under these conditions, we have derived the convergence rates of the estimation error in the Frobenius norm. Numerically, we have proposed a blockwise coordinate descent algorithm that optimizes our objective function by iteratively updating the low-rank component and the diagonal component, and provided both simulated and real data examples showing that our method could have some advantages over a number of alternative estimators. However, this algorithm could lead to a local minimizer instead of a global one. A convenient solution is to initialize from multiple  starting points to increase the chance of finding a global minimizer. We did not 
recommend this, because our deterministic initialization (``warm starts'')  already produced nice results in numerical experiments, and it did not seem 
worthwhile to increase the computational cost.

An immediate extension of our method is that it can be adapted easily to solve the latent variable graphical model selection problem. As mentioned in Section \ref{sec:intro}, \cite{chandrasekaran2010latent} decomposed the observed marginal precision matrix into a sparse and a low-rank component. They used the $\ell_1$-norm as a penalty to encourage sparsity and the nuclear- or trace-norm as a penalty to encourage low-rank-ness. If the rank can be fixed {\it a priori} to be $r$, then we can extend our method easily to solve this problem, by removing the constraint $D\in\mathbf{D}^p$ and adding an $\ell_1$-penalty $\|D\|_1$ to the objective function in (\ref{eq:fr-mle}) instead. If the rank $r$ cannot be fixed, then our rank-penalized method in Section \ref{sec:rank-penalty} can be extended analogously. To solve the modified optimization problem, we only need to modify Algorithm~\ref{algo} slightly by adding an $\ell_1$-penalty on $D$ in  Step 3 to solve for a sparse rather than diagonal component while the low-rank component is fixed. 

Another possible extension could be to consider relaxing the normality assumption in our method. To do so, we would almost certainly need to make explicit assumptions about the tail behavior of the data distribution, which might change the convergence rate of the resulting estimator. Although our objective function is based on the normal likelihood, it works by pushing the covariance matrix estimate towards the sample covariance matrix on one hand and encouraging the assumed ``low-rank + diagonal'' structure on the other. As a result, the estimation accuracy depends on how well the sample covariance matrix can approximate its population counterpart, which is affected by the tail behavior of the data distribution. 

Finally, in this paper we have studied the proposed covariance/precision matrix estimators solely in terms of their estimation accuracy. It could also be interesting to study their performances in other problems, such as discriminant analysis and hypothesis testing, in terms of other performance metrics, such as misclassification probability and statistical power.

\bibliographystyle{apalike}
\bibliography{LDref}


\clearpage
\begin{appendix}
\setcounter{equation}{10}


\subsection*{Proof of Theorem~\ref{thm:fr-Consistency}}
\begin{proof}{Proof}{}%
We use the framework of the proof for the consistency of the sparse precision matrix estimator in \cite{rothman2008sparse}. In spite of the similar framework, our proof is essentially different from theirs in that we are to establish consistency for estimators with the ``low-rank + diagonal'' matrix structure.

To study the solution of the optimization problem (\ref{eq:fr-mle}), we firstly recall the search space,
\[
\mathbf{F}^r=\{\Theta \mid L \in \mathbf{S}^{p,r}_{+}, D \in \mathbf{D}_{++}^p \text{ and } \Theta = -L+D\}.
\]
Base on that, we define another set
\[
\mathbf{E}^r=\{\Delta \mid \Delta=\Theta-\Theta_0, \Theta \in \mathbf{F}^r\},
\]
which can be thought as a ``centered'' version of $\mathbf{F}^r$. As $r \geq r_0$ is assumed in this theorem, we straightforwardly have $\Theta_0 \in \mathbf{F}^r$ and $0 \in \mathbf{E}^r$. 

Let $f(\Theta)=\trace(\Theta S)-\log|\Theta|$ be the value of the objective function at $\Theta$, and $F(\Delta)=f(\Theta_0+\Delta)-f(\Theta_0)$. Let $\hat{\Delta}_r=\hat{\Theta}_r-\Theta_0$, we can prove the desired result
\begin{align}
\label{eq:fr-thm-result}
\|\hat{\Delta}_r\|_F \leq M\max(a_{n,p,r},b_{n,p}),
\end{align}
for some constant $M$, by proving
\begin{align}
\label{eq:F-Delta>0 for Delta in M}
F(\Delta)>F(0)=0 \text{ for all } \Delta \in \mathbf{M}^{2r},
\end{align}
in which
\begin{align*}
\mathbf{M}^{2r}=\mathbf{E}^{2r} \cap \{\Delta \mid \|\Delta\|_F=M\max{(a_{n,p,r},b_{n,p})}\}
\cap \{\Delta \mid \|\Delta\|_{op} \leq C_1\},
\end{align*}
and $C_1$ is a constant so that $\|\hat{\Delta}_r\|_{op} \leq C_1~(r=1,\ldots,p)$. The existence of $C_1$ is validated by Lemma 1.

To clarify this, we show it leads to contradiction if (\ref{eq:F-Delta>0 for Delta in M}) is true while (\ref{eq:fr-thm-result}) is not. As $\|\hat{\Delta}_r\|_F > M\max(a_{n,p,r},b_{n,p})$ and $\|0\|_F < M\max(a_{n,p,r},b_{n,p})$, there exists a real number $0<t<1$ so that $\|(1-t)0+t\hat{\Delta}_r\|_F =M\max(a_{n,p,r},b_{n,p})$. As $\hat{\Delta}_r \in \mathbf{E}^r$ and $0 \in \mathbf{E}^r$, we have $(1-t)0+t\hat{\Delta}_r \in \mathbf{E}^{2r}$. As $\|\hat{\Delta}_r\|_{op}\leq C_1$ by Lemma 1, we have $\|(1-t)0+t\hat{\Delta}_r\|_{op} \leq C_1$. Therefore, $(1-t)0+t\hat{\Delta}_r \in \mathbf{M}^{2r}$ and $F\{(1-t)0+t\hat{\Delta}_r\}>0$ by (\ref{eq:F-Delta>0 for Delta in M}). However, as $\hat{\Delta}_r$ minimizes $F(\Delta)$ and $F(\hat{\Delta}_r)\leq 0$, we also have
\begin{align*}
F\left\{(1-t)0+t\hat{\Delta}_r\right\} \leq (1-t)F(0)+tF(\hat{\Delta}_r) \leq 0
\end{align*}
by convexity of $F(\Delta)$, and this leads to contradiction. 

The remaining work is to prove (\ref{eq:F-Delta>0 for Delta in M}). 

For any $\Delta \in \mathbf{M}^{2r}$, we have
\begin{eqnarray}
\label{eq:fr-F-delta}
F(\Delta)&=&\trace\left\{(\Theta_0+\Delta)S\right\}-\log|\Theta_0+\Delta|-\left\{\trace(\Theta_0 S)-\log|\Theta_0|\right\}\nonumber\\
&=&\trace(\Delta S)-\{\log|\Theta_0+\Delta|-\log|\Theta_0|\}.
\end{eqnarray}
The bound of the second term in (\ref{eq:fr-F-delta}) is irrelevant to the assumed structure of the matrix; according to \cite{rothman2008sparse} and the definition of $\mathbf{M}^{2r}$.
\begin{eqnarray}
\label{eq:logdet expansion}
\log|\Theta_0+\Delta|-\log|\Theta_0|
&\leq &
\trace(\Sigma_0\Delta)-(\|\Theta_0\|_{op}+\|\Delta\|_{op})^{-2}\|\Delta\|_F^2\nonumber\\
& \leq &
\trace(\Sigma_0\Delta)-(c_1^{-1}+C_1)^{-2}\|\Delta\|_F^2.
\end{eqnarray}
We write $C_2=(c_1^{-1}+C_1)^{-2}$. With (\ref{eq:logdet expansion}) plugged into (\ref{eq:fr-F-delta}), we obtain
\begin{eqnarray}
\label{eq:fr-F-delta-2}
F(\Delta) &\geq & C_2\|\Delta\|_F^2+\trace\left\{\Delta(S-\Sigma_0)\right\}.
\end{eqnarray}

Now we derive the bound of $\trace\{\Delta(S-\Sigma_0)\}$ in (\ref{eq:fr-F-delta-2}). We notice that any $\Delta \in \mathbf{E}^{2r}$ can be written as $\Delta=-(L-L_0)+D-D_0$, in which $-(L-L_0) \in \mathbf{S}^{p,3r}$ and $D-D_0 \in \mathbf{D}^p$. By Lemma 2, $\Delta$ can also be decomposed as $\Delta=L_{\Delta}+D_{\Delta}$, so that $L_{\Delta} \in \mathbf{S}^{p,9r}$, $D_{\Delta} \in \mathbf{D}^p$ and $\|\Delta\|_F^2 \geq C_3\left(\|L_{\Delta}\|_F^2+\|D_{\Delta}\|_F^2\right)$ for some constant $C_3$. 
We consider the absolute value,
\begin{eqnarray}
\label{eq:fr-F-delta-tr}
|\trace\left\{\Delta(S-\Sigma_0)\right\}| &\leq &
|\trace\left\{L_{\Delta}(S-\Sigma_0)\right\}|+|\trace\left\{D_{\Delta}(S-\Sigma_0)\right\}|\nonumber\\
&\leq &
\|L_\Delta\|_{*}\|S-\Sigma_0\|_{op}+\|D_{\Delta}\|_{F}\left\{\sum\nolimits_{j=1}^p(s_{jj}-\sigma_{0jj})^2\right\}^{1/2}\nonumber\\
&\leq &
(9r)^{1/2}\|L_{\Delta}\|_F\|S-\Sigma_0\|_{op}+p^{1/2}\|D_\Delta\|_F\underset{1 \leq j \leq p}{\max}|s_{jj}-\sigma_{0jj}|,\nonumber\\
\end{eqnarray}
in which $s_{jj}$ and $\sigma_{0jj}$ are the $j$th diagonal elements in $S$ and $\Sigma_{0}$ respectively. The second inequality is because of the property of dual norm \citep{recht2010guaranteed}. The last inequality uses inequalities regarding different matrix norms \citep{recht2010guaranteed, rothman2008sparse}.

Under the normality assumption, with probability tending to $1$, the sample covariance matrix $S$ satisfies
\begin{align}
\label{eq:S orders}
\underset{1 \leq j \leq p}{\max}|s_{jj}-\sigma_{0jj}| \leq  C_4(\log{p}/n)^{1/2},\quad
\|S-\Sigma_{0} \|_{op} \leq C_4 (p/n)^{1/2},
\end{align}
for some constant $C_4$. The first inequality is by Lemma 1 in \cite{rothman2008sparse}, and the second inequality is by Proposition 2.1 in  \cite{vershynin2012close}.

Combine (\ref{eq:fr-F-delta-tr}) and (\ref{eq:S orders}), we have
\begin{eqnarray}
\label{eq:fr-F-delta-tr-2}
|\trace\left\{\Delta(S-\Sigma_0)\right\}| &\leq &
C_5(\|L_{\Delta}\|_F+\|D_\Delta\|_F)\max(a_{n,p,r},b_{n,p}),
\end{eqnarray}
for some constant $C_5$.

By (\ref{eq:fr-F-delta-2}), (\ref{eq:fr-F-delta-tr-2}) and $\|\Delta\|_F^2 \geq C_3\left(\|L_{\Delta}\|_F^2+\|D_{\Delta}\|_F^2\right)$, 
\begin{eqnarray}
\label{eq:fr-F-delta-final}
F(\Delta) 
&\geq &
C_2\|\Delta\|_F^2-C_5(\|L_{\Delta}\|_F+\|D_\Delta\|_F)\max(a_{n,p,r},b_{n,p})\nonumber\\
&\geq &
C_2\|\Delta\|_F^2-C_5\max(a_{n,p,r},b_{n,p})\left\{2\left(\|L_{\Delta}\|_F^2+\|D_\Delta\|_F^2\right)\right\}^{1/2}\nonumber\\
&\geq &
C_2\|\Delta\|_F^2-C_6\max(a_{n,p,r},b_{n,p})\|\Delta\|_F\nonumber\\
&=&\|\Delta\|_F^2\left\{C_2-C_6\max(a_{n,p,r},b_{n,p})\|\Delta\|_F^{-1}\right\}\nonumber\\
&=&\|\Delta\|_F^2\left(C_2-C_6/M\right)\nonumber\\
&>&0,
\end{eqnarray}
for sufficiently large constant $M$. Constant $C_6$ depends on $C_3$ and $C_5$. 
This completes the proof.
\end{proof}

\subsection*{Proof of Theorem~\ref{thm:r<r0 consistency}}
\begin{proof}{Proof}{}%
Recall that $d_{r,r_0}= \underset{\Theta \in \mathbf{F}^r}{\min}\|\Theta-\Theta_0\|_F$ and $\Theta_r$ is a matrix in $\mathbf{F}^r$ so that $\|\Theta_r-\Theta_0\|_F=d_{r,r_0}$. As
\begin{eqnarray*}
\|\hat{\Theta}_r-\Theta_0\|_F &\leq & \|\hat{\Theta}_r-\Theta_r\|_F +\|\Theta_r-\Theta_0\|_F \\
&=& \|\hat{\Theta}_r-\Theta_r\|_F+d_{r,r_0} \\
&=& \|\hat{\Theta}_r-\Theta_r\|_F+O\left\{\max(a_{n,p,r_0},b_{n,p})\right\},
\end{eqnarray*}
we only need to prove $\|\hat{\Theta}_r-\Theta_r\|_F=O_p\{\max(a_{n,p,r_0},b_{n,p})\}$. 

We use similar technique as in the proof of Theorem~\ref{thm:fr-Consistency}.

Let $f(\Theta)=\trace(\Theta S)-\log|\Theta|$ be the value of the objective function at $\Theta$, and $F_r(\Delta)=f(\Theta_r+\Delta)-f(\Theta_r)$. To obtain the desired result
$\|\hat{\Theta}_r-\Theta_r\|_F\leq M\max(a_{n,p,r_0},b_{n,p})$
for some constant $M$, it is sufficient to prove 
\begin{align}
\label{eq:F-Delta>0 for Delta in M small d}
F_r(\Delta)>F_r(0)=0 \text{ for all } \Delta \in \mathbf{M}^{2r}_r,
\end{align}
in which 
\begin{align*}
\mathbf{M}^{2r}_r
=\{\Delta \mid \Delta=\Theta-\Theta_r, \Theta \in \mathbf{F}^{2r},
\|\Delta\|_F=M\max{(a_{n,p,r_0},b_{n,p})},
\|\Delta\|_{op} \leq C_7\}.
\end{align*}
The constant $C_7$ is defined as follows. As
\[
\|\hat{\Theta}_r-\Theta_r\|_{op} \leq 
\|\hat{\Theta}_r-\Theta_0\|_{op}+\|\Theta_r-\Theta_0\|_F \leq 
C_1+d_{r,r_0},
\]
and $d_{r,r_0} \to 0$, we define $C_7 = 2C_1$ and gaurantee $\|\hat{\Theta}_r-\Theta_r\|_{op} \leq C_7$. Afterwards, the reasoning of the sufficiency of (\ref{eq:F-Delta>0 for Delta in M small d}) is the same as that of the sufficiency of (\ref{eq:F-Delta>0 for Delta in M}), and is omitted.


Now, we prove (\ref{eq:F-Delta>0 for Delta in M small d}).

For any $\Delta \in \mathbf{M}^{2r}_r$, by similar argument as for (\ref{eq:fr-F-delta-2}) and $\|\Theta_r-\Theta_0\|_F=d_{r,r_0}$, with $C_9$ based on $C_7$, we have
\begin{eqnarray}
\label{eq:fr-F-Delta-small-d}
F_r(\Delta) &\geq & C_9\|\Delta\|_F^2+\trace\left\{\Delta(S-\Sigma_r)\right\} \nonumber\\
&=& C_9\|\Delta\|_F^2+\trace\left\{\Delta(S-\Sigma_0)\right\}+\trace\left\{\Delta(\Sigma_0-\Sigma_r)\right\} \nonumber\\
&\geq &C_9\|\Delta\|_F^2+\trace\left\{\Delta(S-\Sigma_0)\right\}-\|\Delta\|_F\|\Sigma_r-\Sigma_0\|_F\nonumber\\
&\geq &C_9\|\Delta\|_F^2+\trace\{\Delta(S-\Sigma_0)\}-C_{10}\|\Delta\|_Fd_{r,r_0},
\end{eqnarray}
for some constant $C_{10}$. The second last inequality is because of Cauchy--Schwarz inequality, and the last inequality uses $\|\Sigma_r-\Sigma_0\|_F =\|\Theta_r^{-1}-\Theta_0^{-1}\|_F \leq C_{10}\|\Theta_r-\Theta_0\|_F$, which can be derived by Taylor expansion. 

By similar argument as from (\ref{eq:fr-F-delta-tr}) to (\ref{eq:fr-F-delta-final}), for $\Delta \in \mathbf{M}^{2r}_r$
\begin{eqnarray}
\label{eq:fr-F-Delta-tr-small-d}
|\trace\{\Delta(S-\Sigma_0)\}| &\leq &
C_{11}\|\Delta\|_F\max(a_{n,p,r_0},b_{n,p}).
\end{eqnarray}

By (\ref{eq:fr-F-Delta-small-d}), (\ref{eq:fr-F-Delta-tr-small-d}) and $d_{r,r_0}=O\left(\max{(a_{n,p,r_0},b_{n,p})}\right)$, with some constant $C_{12}$ based on $C_{10}$ and $C_{11}$,
\begin{eqnarray*}
F_r(\Delta)&\geq &C_9\|\Delta\|_F^2-C_{12}\|\Delta\|_F\max(a_{n,p,r_0},b_{n,p})\nonumber\\
&>&0
\end{eqnarray*}
for sufficiently large $M$. 

This completes the proof.
\end{proof}

\subsection*{Proof of Theorem~\ref{thm:FR-penalty-consistency}}
\begin{proof}{Proof}{}%
Let $f(\Theta)=\trace(\Theta S)-\log{|\Theta|}$, $\hat{\Delta}_r=\hat{\Theta}_r-\Theta_0$ and $F(\hat{\Delta}_r)=f(\hat{\Theta}_r)-f(\Theta_0)$. The objective function in (\ref{eq:rank-MLE}) becomes $f(\hat{\Theta}_r)+\tau(r)$ when $\rank{(L)}$ is fixed to be $r$. 

The discussion in Section \ref{sec:rank-penalty-general} shows that, the convergence rate in Theorem~\ref{thm:FR-penalty-consistency} is already true for $r \in \mathbf{A}_2 \cup \mathbf{A}_3 \cup \{r_0\}$. Thus, if we can prove $f(\hat{\Theta}_r)+\tau(r)>f(\hat{\Theta}_{r_0})+\tau(r_0)$ for all $r \in \mathbf{A}_1 \cup \mathbf{A}_4$ so that these ranks will not be selected, the proof of the theorem will be completed.

For a particular $r \neq r_0$, $\tau(r)$ and $\tau(r_0)$ are both fixed; therefore, all we need is a lower bound of $f(\hat{\Theta}_r)-f(\hat{\Theta}_{r_0})$. We firstly develop a general lower bound, and then discuss $r \in \mathbf{A}_1$ and $r \in \mathbf{A}_4$ separately.

As $f(\Theta_0) \geq f(\hat{\Theta}_{r_0})$, we have 
\[
f(\hat{\Theta}_r)-f(\hat{\Theta}_{r_0}) \geq f(\hat{\Theta}_r)-f(\Theta_0)=F(\hat{\Delta}_r);
\]
and it is sufficient if we have a lower bound for 
\begin{align}
\label{eq:thm4-F-Delta}
F(\hat{\Delta}_r)
=\trace(\hat{\Delta}_rS)-\{\log|\Theta_0+\hat{\Delta}_r|-\log|\Theta_0|\}.
\end{align}

With similar argument as (\ref{eq:logdet expansion}), we have
\begin{eqnarray}
\label{eq:thm4-logdet}
\log|\Theta_0+\hat{\Delta}_r|-\log|\Theta_0|
&\leq &
\trace(\Sigma_0\hat{\Delta}_r)-(\|\Theta_0\|_{op}+\|\hat{\Delta}_r\|_{op})^{-2}\|\hat{\Delta}\|_F^2\nonumber\\
& \leq &
\trace(\Sigma_0\hat{\Delta}_r)-(c_1^{-1}+C_1)^{-2}\|\hat{\Delta}_r\|_F^2.
\end{eqnarray}
Just to clarify, although look alike, the bound of $\|\Delta\|_{op}$ in (\ref{eq:logdet expansion}) is due to the definition of $\mathbf{M}^{2r}$, whereas the bound of $\|\hat{\Delta}_r\|_{op}$ in (\ref{eq:thm4-logdet}) is because $\|\hat{\Delta}_r\|_{op} \leq C_1~(r=1,\ldots,p)$ by Lemma 1.

Plug (\ref{eq:thm4-logdet}) into (\ref{eq:thm4-F-Delta}), we have
\begin{eqnarray}
\label{eq:fr-penalty-F-delta}
F(\hat{\Delta}_r) &\geq & C_{2}\|\hat{\Delta}_r\|_F^2+\trace \{\hat{\Delta}_r(S-\Sigma_0)\}.
\end{eqnarray}

Let $\hat{L}_{r}$ and $\hat{D}_r$ be the low-rank matrix component and diagonal matrix component of $\hat{\Theta}_r$ respectively, we have $\hat{\Delta}_r=-(\hat{L}_{r}-L_0)+(\hat{D}_r-D_0)$, in which $-(\hat{L}_{r}-L_0) \in \mathbf{S}^{p,r+r_0}$ and $\hat{D}_r-D_0 \in \mathbf{D}^p$. By Lemma 2, $\hat{\Delta}_r$ can also be written as $\hat{\Delta}_r=L_{\hat{\Delta}_r}+D_{\hat{\Delta}_r}$, in which $L_{\hat{\Delta}_r} \in \mathbf{S}^{p,3(r+r_0)}$, $D_{\hat{\Delta}_r} \in \mathbf{D}^p$ and $\|\hat{\Delta}_r\|^2_F \geq C_3(\|L_{\hat{\Delta}_r}\|_F^2+\|D_{\hat{\Delta}_r}\|_F^2)$.  

By similar argument as (\ref{eq:fr-F-delta-tr}) -- (\ref{eq:fr-F-delta-tr-2}), the second part in (\ref{eq:fr-penalty-F-delta}) can be bounded as 
\begin{eqnarray}
\label{eq:fr-penalty-F-delta-tr}
|\trace\{\hat{\Delta}_r(S-\Sigma_0)\}| &\leq &
\left\{3(r+r_0)\right\}^{1/2}\|L_{\hat{\Delta}_r}\|_F\|S-\Sigma_0\|_{op}\nonumber\\
&&+p^{1/2}\|D_{\hat{\Delta}_r}\|_F\underset{1\leq j \leq p}{\max}|s_{jj}-\sigma_{0jj}|\nonumber\\
&\leq &
C_{14}\|\hat{\Delta}_r\|_F\max{\{a_{n,p,(r+r_0)},b_{n,p}\}}.
\end{eqnarray}
for some constant $C_{14}$.

Plug (\ref{eq:fr-penalty-F-delta-tr}) into (\ref{eq:fr-penalty-F-delta}), we have
\begin{eqnarray}
\label{eq:fr-penalty-F-delta-2}
F(\hat{\Delta}_r) &\geq & 
C_2\|\hat{\Delta}_r\|_F^2
-C_{14}\|\hat{\Delta}_r\|_F\max{\{a_{n,p,(r+r_0)},b_{n,p}\}}.
\end{eqnarray}

With the general lower bound of $F(\hat{\Delta}_r)$ obtained, we now consider $r \in \mathbf{A}_1$.

When $r \in \mathbf{A}_1$, as $r < r_0$, we replace the $a_{n,p,(r+r_0)}$ in (\ref{eq:fr-penalty-F-delta-2}) with $a_{n,p,r_0}$, and obtain
\begin{eqnarray}
\label{eq:fr-penalty-F-delta-3}
F(\hat{\Delta}_r) &\geq & 
C_2\|\hat{\Delta}_r\|_F^2
-C_{15}\|\hat{\Delta}_r\|_F\max{(a_{n,p,r_0},b_{n,p})},
\end{eqnarray}
for some constant $C_{15}$. By the definition of $\mathbf{A}_1$, we can represent $d_{r,r_0}$ as \[d_{r,r_0}=\eta_{n,p,r_0}\max(a_{n,p,r_0},b_{n,p})\] for some $\eta_{n,p,r_0} \to \infty$. By the definition of the distance $d_{r,r_0}$, we have $\|\hat{\Delta}_r\|_F \geq d_{r,r_0}$. With these facts, (\ref{eq:fr-penalty-F-delta-3}) can be simplified as 
\begin{eqnarray}
\label{eq:fr-penalty-F-delta-r<r0}
F(\hat{\Delta}_r) &\geq & 
\|\hat{\Delta}_r\|_F^2\left\{
C_2-C_{15}\|\hat{\Delta}_r\|_F^{-1}\max(a_{n,p,r_0},b_{n,p})
\right\}\nonumber\\
&\geq&
\|\hat{\Delta}_r\|_F^2\left(
C_2-C_{15}\eta_{n,p,r_0}^{-1}
\right)\nonumber\\
&\geq&
C_2\|\hat{\Delta}_r\|_F^2/2\nonumber\\
&\geq & C_2d_{r,r_0}^2/2,
\end{eqnarray}
when $n$ and $p$ are sufficiently large.

By (\ref{eq:fr-penalty-F-delta-r<r0}) and Condition \ref{(C.3)}, we have
\begin{eqnarray}
\label{eq:fr-penalty-final-r<r0}
&&\left\{f(\hat{\Theta}_r)+\tau(r)\right\}-\left\{f(\hat{\Theta}_{r_0})+\tau(r_0)\right\}\nonumber\\
& \geq & 
C_2d_{r,r_0}^2/2+\tau(r)-\tau(r_0)\nonumber \\
&> & 0,
\end{eqnarray}
when $n$ and $p$ are sufficiently large.

When $r \in \mathbf{A}_4$, the $a_{n,p,(r+r_0)}$ in (\ref{eq:fr-penalty-F-delta-2}) can be replaced with $a_{n,p,r}$, and we obtain
\begin{eqnarray*}
\label{eq:fr-penalty-F-delta-4}
F(\hat{\Delta}_r) &\geq & 
C_2\|\hat{\Delta}_r\|_F^2
-C_{15}\|\hat{\Delta}_r\|_F\max{(a_{n,p,r},b_{n,p})}.
\end{eqnarray*}
As $\mathbf{A}_4$ is defined so that $r/\max(r_0,\log{p})\to \infty$, we have $a_{n,p,r}/b_{n,p} \to \infty$ and
\begin{eqnarray}
\label{eq:fr-penalty-F-delta-5}
F(\hat{\Delta}_r) &\geq & 
C_2\|\hat{\Delta}_r\|_F^2
-C_{15}\|\hat{\Delta}_r\|_Fa_{n,p,r}.
\end{eqnarray}
The right hand side of the inequality in (\ref{eq:fr-penalty-F-delta-5}) is quadratic in $\|\hat{\Delta}_r\|_F$ and can be minimized analytically. Thus, (\ref{eq:fr-penalty-F-delta-5}) is bounded as
\begin{eqnarray}
\label{eq:fr-penalty-F-delta-r>r0}
F(\hat{\Delta}_r) &\geq & -C_{16}a_{n,p,r}^2,
\end{eqnarray}
in which $C_{16}$ is some positive constant based on $C_2$ and $C_{15}$. ,

By (\ref{eq:fr-penalty-F-delta-r>r0}) and Condition \ref{(C.4)}, we have
\begin{eqnarray}
\label{eq:fr-penalty-final-r>r0}
&&\left\{f(\hat{\Theta}_r)+\tau(r)\right\}-\left\{f(\hat{\Theta}_{r_0})+\tau(r_0)\right\}\nonumber\\
& \geq & 
-C_{16}a_{n,p,r}^2+\tau(r)-\tau(r_0)\nonumber \\
&> & 0,
\end{eqnarray}
when $n$ and $p$ are sufficiently large.

Results (\ref{eq:fr-penalty-final-r<r0}) and (\ref{eq:fr-penalty-final-r>r0}) together complete the proof.
\end{proof}

\section*{Lemmas and proof of lemmas}

This part of the supplementary material contains some lemmas and their proofs. Lemma 1 and Lemma 2 are repeatedly used in the proof of Theorem~\ref{thm:fr-Consistency} -- Theorem~\ref{thm:FR-penalty-consistency}; Lemma 3 is a useful result for the proof of Lemma 2; Lemma 4 is used to justify Algorithm \ref{algo}. 

\begin{lemma}{Lemma 1.}{}
Let $\hat{\Theta}_r$ be the solution of the low-rank and diagonal matrix decomposition when the rank is fixed to be $r$,
\begin{eqnarray}
\hat{\Theta}_r&=&\underset{\Theta}{\argmin}\{\trace(\Theta S)-\log|\Theta|\},\nonumber\\
\text{subject to}&& \Theta=-L+D,~
\Theta \in \mathbf{S}^{p}_{+},~
L \in \mathbf{S}^{p,r}_{+}, ~
D \in \mathbf{D}^p,
\label{eq:bound operator norm}
\end{eqnarray}
in which $S$ is the sample covariance matrix, we have $\|\hat{\Theta}_r-\Theta_0\|_{op}<C$ for some constant $C$, with probability tending to $1$.
	\label{lemma:HatThetar-Theta}
\end{lemma}

\begin{proof}{Proof}{}
In the following proof, we will use the fact that, with probability tending to $1$,
\begin{eqnarray}
\lambda_{\max}(S^{-1})=\lambda_{\min}^{-1}(S) &\leq& \left\{\lambda_{\min}(\Sigma_0)-c(p/n)^{1/2}\right\}^{-1}\nonumber\\
&\leq & 2/c_1,
\label{eq:bound inv(S)}
\end{eqnarray}
for some constants $c$ and $c_1$, where $c_1$ has been defined in Condition \ref{(C.1)}. 

To prove this lemma,  it suffices to show that 
\begin{align}
\lambda_{\max}(\hat{\Theta}_r) \leq \lambda_{\max}(S^{-1}).
\label{eq:op bound change target}
\end{align}
This is because
\begin{eqnarray*}
\|\hat{\Theta}_r-\Theta_0\|_{op} &\leq&
\|\hat{\Theta}_r-S^{-1}\|_{op}+\|S^{-1}-\Theta_0\|_{op}\\
&\leq& \max\left\{\lambda_{\max}(\hat{\Theta}_r),\lambda_{\max}(S^{-1})\right\}+\max\left\{\lambda_{\max}(S^{-1}),\lambda_{\max}(\Theta_0)\right\}\nonumber\\
&\leq & 4/c_1,
\end{eqnarray*}
The second inequality is due to the fact that $\hat{\Theta}_r$, $S^{-1}$ and $\Theta_0$ are all positive definite. The last inequality is because of  (\ref{eq:op bound change target}) and (\ref{eq:bound inv(S)}).




It remains to show (\ref{eq:op bound change target}). We will prove that, if $\lambda_{\max}(\Theta)>\lambda_{\max}(S^{-1})$ instead (i.e. (\ref{eq:op bound change target}) isn't true), then $\Theta$ must not be the solution to (\ref{eq:bound operator norm}) because the objective function in (\ref{eq:bound operator norm}) can always be further decreased. We conduct this proof in two steps. 

Step 1: If $\lambda_{\max}(\Theta)>\lambda_{\max}(S^{-1})$, the objective function cannot reach its minimum.

Let $\Theta= D-L$ in which $D$ and $L$ are constrained as in (\ref{eq:bound operator norm}). We eigendecompose $\Theta$ as
\[
\Theta = D-L=T\Lambda T^{\T},
\]
in which $T=\left(t_1,\ldots,t_p\right)$ and $\Lambda=diag(\lambda_1,\ldots,\lambda_p)$. Without loss of generality, let the eigenvalues be aligned in descending order.  With basic calculus, the objective function in (\ref{eq:bound operator norm}) can be rewritten as
\begin{align}
\trace(\Lambda T^{\T}ST)-\log|\Lambda|
=\sum_{j=1}^p\left(\lambda_jt_j^{\T}S t_j-\log\lambda_j\right), 
\label{eq:B4-3}
\end{align}
for which the partial differentiation with respect to $\lambda_1$ is $t_1^{\T}S t_1-\lambda_1^{-1}$. Hence, due to convexity, (\ref{eq:B4-3}) may reach its minimum when $\lambda_1=(t_1^{\T}S t_1)^{-1}$. However,  $\lambda_1>(t_1^{\T}S t_1)^{-1}$ strictly because
\[\lambda_1=\lambda_{\max}(\Theta)>\lambda_{\max}(S^{-1})=\lambda_{\min}^{-1}(S)\geq (t_1^{\T}S t_1)^{-1}.\]
Therefore, (\ref{eq:B4-3}) cannot reach its minimum. 

Step 2: Given that $\lambda_1>(t_1^{\T}S t_1)^{-1}$, the objective function can be further decreased if (not only if) we change the $D$ (in $\Theta = D-L$) in a way that both $t_1^{\T}S t_1$ and $\sum_{j=2}^{p}\left(\lambda_j t_j^{\T}S t_j-\log\lambda_j\right)$ remain unchanged but $\lambda_1$ decreases.

We now show that such a change in $D$ does exist.  By employing the results of differentiating eigenvalues and eigenvectors in \cite{magnus1985differentiating}, we have the following three results. First of all
\begin{eqnarray}
d \lambda_1 &=&t_1^{\T}(dD)t_1,
\label{eq:B4-6}
\end{eqnarray}
Secondly, 
\begin{eqnarray}
d\left(t_1^{\T}St_1\right)&=&
2(St_1)^{\T}dt_1\nonumber\\
&=&2(St_1)^{\T}(\lambda_1I_p-\Theta)^+(dD)t_1.
\label{eq:B4-4}
\end{eqnarray}
Lastly, 
\begin{eqnarray}
d\left\{\sum_{j=2}^{p}\left(\lambda_j t_j^{\T}S t_j-\log\lambda_j\right)\right\}&=& \sum_{j=2}^{p}(t_j^{\T}S t_j-\lambda_j^{-1})d\lambda_j
+2\lambda_j(S t_j)^{\T}d t_j \nonumber\\
&=&
\sum_{j=2}^{p}(t_j^{\T}S t_j-\lambda_j^{-1}) t_j^{\T}(dD) t_j \nonumber\\
&&+2\lambda_j(S t_j)^{\T}(\lambda_jI_p-\Theta)^+(dD) t_j,
\label{eq:B4-5}
\end{eqnarray}
in which $dD$ is a diagonal matrix representing an infinitesimal change of $D$ and $(\cdot)^+$ is the Moore-Penrose inverse. Expressions (\ref{eq:B4-6}),(\ref{eq:B4-4}) and (\ref{eq:B4-5}) are all linear with respect to the elements in $dD$ and $t_1 \neq 0$ obviously. Hence, we can surely solve $dD$ from setting (\ref{eq:B4-4}) and (\ref{eq:B4-5}) to be $0$ and (\ref{eq:B4-6}) to be negative. 

In summary, we have shown that if we change $D$ by $dD$, the objective function (\ref{eq:B4-3}) decreases. Therefore, we have proved that $\Theta$ is not the solution to  (\ref{eq:bound operator norm}). This completes the proof.
\end{proof}

\begin{lemma}{Lemma 2.}{}
If a $p \times p$ matrix $M$ can be written as $M=L+D$, in which $L \in \mathbf{S}^{p,r}$ and $D \in \mathbf{D}^p$, then $M$ can also be written as $M=L'+D'$, in which $L' \in \mathbf{S}^{p,3r}$, $D' \in \mathbf{D}^p$ and 
\begin{align*}
\|M\|_F^2 \geq C\left(\|L'\|_F^2+\|D'\|_F^2\right),
\end{align*}
for some positive constant $C$.
\label{lemma:fro-norm-inequality}
\end{lemma}

\begin{proof}{Proof}{}
Let $M_{ij}$ and $L_{ij}$ be the entries in the $i$th row and $j$th column of $M$ and $L$ respectively; let $D_{j}$ be the $j$th diagonal entry of D. Similarly, $L'_{ij}$ and $D_{j}'$ are defined. Define the index set $\mathbf{B}=\{j:L_{jj}^2>\sum_{i \neq j}L_{ij}^2\}$. 

According to Lemma 3, the cardinality of $\mathbf{B}$ is at most $2r-1$. We set $L'_{jj}=M_{jj}/2$ for $j \in \mathbf{B}$ and $L'_{ij}=L_{ij}$ for $i \neq j$ and $i=j \notin \mathbf{B}$; $D'$ is set accordingly so that $M=L'+D'$. As at most $2r-1$ diagonal entries of $L'$ are different from those of $L$, $\rank{(L')}<3r$. Now we prove $\|M\|_F^2 \geq C\left(\|L'\|_F^2+\|D'\|_F^2\right)$ for some constant $C$.

We notice two properties: (1) for $j \in \mathbf{B}$, $(L'_{jj})^2=M_{jj}^2/4$; (2) for $j \notin \mathbf{B}$, $(L'_{jj})^{2} \leq \sum_{i \neq j}(L'_{ij})^2=\sum_{i \neq j}M_{ij}^2$. As a result, 

\begin{eqnarray}
\|M\|_{F}^2&=&\sum_{j=1}^pM_{jj}^2+\sum_{j=1}^p\sum_{i\neq j}M_{ij}^2 \nonumber\\
&\geq & 4\sum_{j \in \mathbf{B}}(L'_{jj})^2+\sum_{j=1}^p\sum_{i\neq j}(L'_{ij})^2 \nonumber\\
&\geq & 1/2\left\{
\sum_{j \in \mathbf{B}}(L'_{jj})^2
+2\sum_{j=1}^p\sum_{i\neq j}(L'_{ij})^2
\right\} \nonumber\\
&\geq & 1/2\left\{\sum_{j \in \mathbf{B}}(L'_{jj})^2+\sum_{j \notin \mathbf{B}}(L_{jj}')^2+\sum_{j=1}^p\sum_{i \neq j}(L_{ij}')^2\right\} \nonumber\\
&=&1/2\|L'\|_F^2.
\label{eq:L vs M}
\end{eqnarray}
The first inequality is because of property (1) and the third inequality is because of property (2).

Finally, by (\ref{eq:L vs M}) and $\|D'\|_F  \leq  \|M\|_F + \|L'\|_F$, we have
\begin{eqnarray*}
\|D'\|_F^2 & \leq & 2(\|M\|_F^2 + \|L'\|^2_F)\\
&\leq & 6\|M\|_F^2,
\end{eqnarray*}
we have $\|L'\|_F^2+\|D'\|_F^2 \leq 8\|M\|_F^2$
and 
$
\|M\|_F^2 \geq C(\|L'\|_F^2+\|D'\|_F^2)
$
for $C=1/8$.
This completes the proof.
\end{proof}

\begin{lemma}{Lemma 3.}{}
Let $A$ be a $p \times p$ matrix with $\rank(A) = r$ ($r \leq p$) and $a_{ij}$ be the element in the $i$th row and $j$th column, the number of column vectors in $A$ that satisfy $a_{jj}^2>\sum_{i \neq j}a_{ij}^2$ is at most $2r-1$.
\label{lemma:l2 diagonal dominance}
\end{lemma}

\begin{proof}{Proof}{}
Let $a_j$ be the $j$th column vector in $A$. If it satisfies $a_{jj}^2>\sum_{i \neq j}a_{ij}^2$, we say this column is diagonally dominant and is dominated by the $j$th element. Let $\mathbf{R}^p$ denote the dimension $p$ vector space, and $\mathbf{R}^{p,r}$ denote the column space of $A$. Straightforwardly, $\mathbf{R}^{p,r}$ is a subspace of $\mathbf{R}^p$ that contains at most $r$ linearly independent vectors.

\emph{Finding out the upper bound of the number of diagonally dominant column vectors in $A$ is equivalent to considering at most how many vectors in $\mathbf{R}^{p,r}$ can be dominated by one of its entries}. The equivalence requires, when we count in $\mathbf{R}^{p,r}$, if two vectors are dominated by the same entry (e.g. $j$th), they are counted as one vector. Now, we count in $\mathbf{R}^{p,r}$.

Without loss of generality, we assume the first $r$ columns ($a_1,\ldots,a_r$) in $A$ are orthogonal to each other and are unit vectors. This is valid because for any given $A$, without changing the column space, we can (1) change the order of the columns by moving $r$ linearly independent column vectors to the left and (2) orthonormalize these linearly independent vectors.

Let 
\[
V_{p \times r}=\left(
\begin{matrix}
a_1,\ldots,a_r
\end{matrix}
\right)
=\left(
\begin{matrix}
b_1^{\T}\\
\vdots\\
b_p^{\T}
\end{matrix}
\right),
\]
in which $b_1,\ldots,b_p$ are $r \times 1$ vectors. Any vector in $\mathbf{R}^{p,r}$ can be written as $V_{p \times r}k$ where $k$ is a $r \times 1$ vector; therefore, a vector dominated by the $j$th element can be in $\mathbf{R}^{p,r}$ if and only if there is a vector $k \neq 0$ and 
\[
(b_j^{\T}k)^2>\sum_{i \neq j}(b_i^{\T}k)^2.
\]
The inequality is equivalent to 
\[
k^{\T}b_jb_j^{\T}k>\sum_{i \neq j}k^{\T}b_ib_i^{\T}k,
\]
and
\[
k^{\T}(V^{\T}V-2b_jb_j^{\T})k<0.
\]

The existence of $k$ suggests $V^{\T}V-2b_jb_j^{\T}$ has negative eigenvalues. As $V$ consists of orthonormal vectors, we conclude the smallest eigenvalue of 
\[
V^{\T}V-2b_jb_j^{\T}=
I_{r}
-2b_jb_j^{\T}
\]
must be negative.

Let $\lambda_{\min}(\cdot)$ be the smallest eigenvalue of a matrix and $u$ be the corresponding eigenvector of $\lambda_{\min}\left(I_{r}-2b_jb_j^{\T}\right)$. We have
\begin{eqnarray*}
\lambda_{\min}\left(I_{r}-2b_jb_j^{\T}\right) &
=&u^{\T}\left(I_{r}-2b_jb_j^{\T}\right)u\\
&=&1-2(u^{\T}b_j)^2\\
&\geq &1-2\|b_j\|^2
\end{eqnarray*}
and consequently $\|b_j\|^2>1/2$. Finally, noticing $\sum_{i=1}^p\|b_i\|^2=\sum_{j=1}^r\|a_j\|^2=r$, we conclude there are at most $2r-1$ $b_j$ with $\|b_j\|^2>1/2$.
\end{proof}

\begin{lemma}{Lemma 4.}{}
When $D$ is fixed and positive definite, the objective function 
\[\trace\left\{(D-L)S\right\}-\log|D-L|\]
can be minimized with respect to $L$ analytically. 
 
Eigen-decompose $D^{1/2}SD^{1/2}$, let $w_1,\ldots,w_p$ be the eigenvalues in descending order and $u_1,\ldots,u_p$ be the associated eigenvectors. Let $U=(u_1,\ldots,u_r)$, $V=\diag\{1-1/\max{(w_1,1)},$ $\ldots,1-1/\max{(w_r,1)}\}$, then
$L=D^{1/2}U V U^{\T}D^{1/2}$ is the analytic solution.
\label{lemma:fix D solve L}
\end{lemma}

\begin{proof}{Proof}{}
Since $D$ and $S$ are fixed, the target can be simplified as maximizing
\begin{align*}
&\trace(LS)+\log|D-L|\\
=&
\trace\left\{(D^{-1/2}LD^{-1/2})D^{1/2}SD^{1/2}\right\}+\log|I_p-D^{-1/2}LD^{-1/2}|+\log|D|.
\end{align*}

Let the low-rank part be eigen-decomposed as $D^{-1/2}LD^{-1/2}=\tilde{U}\tilde{V}\tilde{U}^{\T}$, in which $\tilde{U}=(\tilde{u}_1,\ldots,\tilde{u}_r)$ is a $p\times r$ matrix and $\tilde{V}=\diag(\tilde{v}_1,\ldots,\tilde{v}_r)$ is a $r \times r$ diagonal matrix. Also, without of loss generality, let $\tilde{v}_1,\ldots,\tilde{v}_r$ be in descending order. Then, we need to maximize
\begin{align}
\label{DL_eq:fix D solve L}
\trace\left\{\tilde{V}\tilde{U}^{\T}D^{1/2}SD^{1/2}\tilde{U}\right\}+\log|I_r-\tilde{V}|.
\end{align}

Regardless of $\tilde{V}$, We have 
\begin{eqnarray*}
\trace\left\{\tilde{V}\tilde{U}^{\T}D^{1/2}SD^{1/2}\tilde{U}\right\}
&\leq &\sum_{i=1}^r \lambda_i(\tilde{V}) \lambda_i(\tilde{U}^{\T}D^{1/2}SD^{1/2}\tilde{U})\nonumber\\
&\leq & \sum_{i=1}^r \tilde{v}_i \lambda_i(D^{1/2}SD^{1/2})\nonumber\\
&=& \sum_{i=1}^r \tilde{v}_i w_i,
\end{eqnarray*}
where $\lambda_i(\cdot)$ is the $i$th largest eigenvalue of the input matrix. The first and second inequalities follow Theorem 3.34 and Theorem 3.19 in \cite{schott2005matrix} respectively. The maximum can be achieved when $\tilde{U} = U$.

When $\tilde{U} = U$, maximizing (\ref{DL_eq:fix D solve L}) is equivalent to maximizing 
\begin{align*}
\tilde{v}_iw_i+\log(1-\tilde{v}_i) \quad  (i = 1,\ldots,r),
\end{align*}
subject to $\tilde{v}_i \in [0,1)$. By basic calculus, we need $\tilde{v}_i = 1-1/\max(w_i,1)$.
\end{proof}


\end{appendix}

\CJShistory
\end{document}